%% file: main.tex
\documentclass[12pt]{article}
\input{header}

\title{Finding Most Shattering Minimum Vertex Cuts of Polylogarithmic Size in Near-Linear Time}

\author{
    Kevin Hua\thanks{
        University of Michigan,
        \tt{kevinhua@umich.edu}
    } \and
    Daniel Li\thanks{
        University of Michigan,
        \tt{dejl@umich.edu}
    } \and
    Jaewoo Park\thanks{
        University of Michigan,
        \tt{pjaewoo@umich.edu}
    } \and
    Thatchaphol Saranurak\thanks{
        University of Michigan,
        \tt{thsa@umich.edu}. Supported by NSF grant CCF-2238138.
    }
}

\date{}
\begin{document}

\pagenumbering{gobble}
\maketitle
\begin{abstract}
    We show the first near-linear time randomized algorithms for listing all minimum vertex cuts of polylogarithmic size that separate the graph into at least \emph{three} connected components (also known as \emph{shredders}) and for finding the \emph{most shattering} one, i.e., the one maximizing the number of connected components. Our algorithms break the quadratic time bound by Cheriyan and Thurimella (STOC'96) for both problems that has stood for more than two decades. Our work also removes a bottleneck to near-linear time algorithms for the vertex connectivity augmentation problem (Jordan '95). Note that it is necessary to list only minimum vertex cuts that separate the graph into at least three components because there can be an exponential number of minimum vertex cuts in general.

    To obtain near-linear time algorithms, we have extended techniques in local flow algorithms developed by Forster et al.\ (SODA'20) to list shredders on a local scale. We also exploit fast queries to a pairwise vertex connectivity oracle subject to vertex failures (Long and Saranurak FOCS'22, Kosinas ESA'23). This is the first application of connectivity oracles subject to vertex failures to speed up a static graph algorithm.
\end{abstract}

\newpage
{
    \hypersetup{linkcolor=DarkRed}
    \tableofcontents
}

\newpage
\pagenumbering{arabic}
\allowdisplaybreaks

\input{src/1,2-intro,overview}
\input{src/3,4-prelim,CT99}
\input{src/6-local}
\input{src/7-unverified}
\input{src/8,9-capture,listing}
\input{src/10-shattering}

\newpage
\bibliographystyle{alpha}
\bibliography{citations}

\newpage
\appendix
\input{src/A-appendix}

\end{document}

%% file: header.tex
\usepackage{fullpage}
\usepackage{graphicx}
\usepackage[dvipsnames]{xcolor}
\usepackage{caption}
\usepackage{enumerate}

\usepackage{amsmath, amsthm, amssymb}
\usepackage{thmtools, thm-restate}

\usepackage{tikz}
\usetikzlibrary{shapes.geometric, arrows, positioning, calc}

\definecolor{ForestGreen}{rgb}{0.1333,0.5451,0.1333}
\definecolor{DarkRed}{rgb}{0.65,0,0}
\definecolor{Red}{rgb}{1,0,0}
\usepackage[
    linktocpage=true,
    pagebackref=true,
    colorlinks,
    linkcolor=DarkRed,
    citecolor=ForestGreen,
    bookmarks,
    bookmarksopen,
    bookmarksnumbered,
]{hyperref}
\usepackage{cleveref}
\usepackage[T1]{fontenc}

\declaretheorem[numberwithin=section]{theorem}
\declaretheorem[numberlike=theorem, name=Theorem]{thm}

\declaretheorem[numberlike=theorem, name=Lemma]{lem}
\Crefname{lem}{Lemma}{Lemmas}
\declaretheorem[numberlike=theorem]{fact}

\declaretheorem[numberlike=theorem, style=definition]{definition}
\Crefname{def}{Definition}{Definitions}
\declaretheorem[numberlike=theorem, style=definition]{remark}

\usepackage{algorithm}
\usepackage{algpseudocodex}

\newcommand{\INDSTATE}{\Statex\hspace{\algorithmicindent}}

\newcommand{\sm}{\setminus}
\newcommand{\e}{{\rightsquigarrow}}

\newcommand{\poly}{\text{poly}}
\newcommand{\polylog}{\text{polylog}}
\renewcommand{\tt}{\texttt}
\renewcommand{\O}{\mathcal{O}}
\newcommand{\tO}{\tilde{\mathcal{O}}}

\newcommand{\vol}{\mathrm{vol}}
\newcommand{\mc}[1]{\mathcal{#1}}
\newcommand{\R}{\mc{R}}
\newcommand{\B}{\mc{B}}
\newcommand{\C}{\mc{C}}

\renewcommand{\L}{\mc{L}}
\newcommand{\U}{\mc{U}}
\newcommand{\I}{\mc{I}}
\renewcommand{\C}{\mc{C}}
\newcommand{\TRU}{\tt{true}}
\newcommand{\FLS}{\tt{false}}

\newcommand{\SHR}{\tt{Shredders}}
\newcommand{\AKS}{\tt{All-$k$-shredders}}
\newcommand{\LVC}{\tt{LocalVC}}
\newcommand{\fmt}{\texorpdfstring}

%% file: src/1,2-intro,overview.tex
\section{Introduction}
Given an $n$-vertex, $m$-edge, undirected graph $G$, a \emph{minimum vertex cut} is a set of vertices whose removal disconnects $G$ of smallest possible size. The \emph{vertex connectivity} of $G$ is the size of any minimum vertex cut. The problem of efficiently computing vertex connectivity and finding a corresponding minimum vertex cut has been extensively studied for more than half a century \cite{Kleitman1969methods, Podderyugin1973algorithm, EvenT75, Even75, Galil80, EsfahanianH84, Matula87, BeckerDDHKKMNRW82, LinialLW88, CheriyanT91, NI92, CheriyanR94, Henzinger97, HenzingerRG00, Gabow06, Censor-HillelGK14}. Let $k$ denote the vertex connectivity of $G$. Recently, an $\O(k^2m)$ time algorithm was shown in \cite{FNSYY20}, which is near-linear when $k = \O(\polylog (n))$. Then, an almost-linear $m^{1+o(1)}$ time algorithm for the general case was finally discovered \cite{LNPSY21,CKLPGS22}. In this paper, we show new algorithms for two closely related problems:
\begin{enumerate}
    \item \label{enu:prob-shattering} Find a \emph{most-shattering} minimum vertex cut, i.e., a minimum vertex cut $S$ such that the number of (connected) components of $G\setminus S$ is maximized over all minimum vertex cuts.
    \item \label{enu:prob-shredder} List all minimum vertex cuts $S$ such that $G \setminus S$ has at least three components.
\end{enumerate}
In the latter problem, the restriction to at least three components is natural for polynomial time algorithms. This is because there are at most $n$ many minimum vertex cuts whose removal results in at least three components \cite{Jor99}, while the total number of minimum vertex cuts can be exponential or, more specifically, at least $2^{k}(n/k)^{2}$ \cite{PY21}.\footnote{Both problems are specific to vertex cuts; recall that every minimum \emph{edge} cut always separates a graph into two components.} Before we discuss the literature on both problems, let us introduce some terminology. We say that a vertex set $S$ is a \emph{separator} if $G \sm S$ is not connected and a \emph{shredder} if $G \sm S$ has at least three components. An $s$-separator ($s$-shredder) is a separator (shredder) of size $s$. In other words, the second problem is to list all $k$-shredders.\footnote{For example, in a tree, vertices with degree at least three are 1-shredders. In the complete bipartite graph $K_{k,k}$ with $k\geq3$, both the left and right halves of the bipartition are $k$-shredders.}

\paragraph{History.}
Most shattering minimum vertex cuts, or most shattering min-cuts for short, have played an important role in the \emph{vertex connectivity augmentation} problem. In this problem, we wish to determine the minimum number of edges required to augment the vertex connectivity of $G$ by one. Let $b(G)$ denote the number of components in $G \setminus S$ when $S$ is a most shattering min-cut. Jordán \cite{Jor95} showed that $b(G)$ is a lower bound for the optimal solution of the augmentation problem. He asked whether there exist efficient algorithms for computing $b(G)$. Without computing $b(G)$ explicitly, he gave an approximation algorithm based on $b(G)$ and another quantity omitted here. 

The first efficient algorithm for computing most shattering min-cuts (Problem~\ref{enu:prob-shattering}) was by Cheriyan and Thurimella \cite{CT99}. They showed an algorithm with $\O(k^2n^2+k^3n^{3/2})$ running time and used it as a subroutine to speed up Jordán's approximation algorithm mentioned above to $\O(\min(k,\sqrt{n})k^{2}n^{2}+kn^{2}\log n)$ time.\footnote{Years later, Végh \cite{Veg10} gave the first exact polynomial time algorithm that runs in $\O(kn^{7})$ time using a different approach.} A more general problem was also studied: given an arbitrary parameter $k'$, we must determine the minimum number of edges required to increase the vertex connectivity of $G$ from $k$ to $k'$. Although polynomial time algorithms are still not known for this problem, Jackson and Jordán \cite{JJ05} showed a fixed-parameter-tractable algorithm that runs in $\O(n^{6}+c_{k}n^{3})$ time, where $c_{k}$ is an exponential value that depends only on $k$. The parameter $b(G)$ is also important in both their algorithm and analysis. 

Cheriyan and Thurimella's algorithm also lists all $k$-shredders in the same time, effectively solving Problem \ref{enu:prob-shredder}. Subsequently, much progress has been made on bounding the number of $k$-shredders. Jordan \cite{Jor99} showed that, for any $k$, the number of $k$-shredders in any $k$-vertex-connected graph is at most $n$. This bound was analyzed and improved for specific values of $k$ in a long line of work \cite{Tsu03, LN07, Ega08, EOT08, Hir10, Tak20, TH20, HTH20}. Faster algorithms for listing $k$-shredders for small values of $k$ are known. For $k=1$, we can list all 1-shredders by listing all articulation points of $G$ with degree at least three. This can be done in linear time, as shown by Tarjan \cite{Tar74}. For $k=2$, listing all 2-shredders can be done in linear time after decomposing $G$ into triconnected components, as shown by Hopcroft and Tarjan \cite{HT73}. For $k=3$, listing all 3-shredders can be done in linear time, as shown by Hegde \cite{Heg06}.

Hegde's motivation behind his algorithm is a surprising connection to the problem of finding an even-length cycle in a directed graph. Vazirani and Yannakakis \cite{VY90} showed that this problem is equivalent to many other fundamental questions, e.g., checking if a bipartite graph has a Pfaffian orientation (see \cite{Kas63,Kas67}) and more. Via the work of \cite{RST99}, Hegde observed that to speed up the algorithm for finding Pfaffian orientations in bipartite graphs beyond $\O(n^{2})$ time, one needs to quickly determine whether a special type of bipartite graph contains a 4-shredder.\footnote{It turns out that this special type of bipartite graph (defined as a \emph{brace} by \cite{RST99}) is 3-connected when it contains at least five vertices, but not necessarily 4-connected. To truly remove this bottleneck, one must at the very least determine whether a 3-connected brace contains a 4-shredder. Therefore, our algorithm does not remove this bottleneck. In the previous version of this paper, we misunderstood Hegde's statement and incorrectly claimed that we removed this bottleneck.} Up until now, the fastest known algorithm for listing 4-shredders in 4-connected graphs is by \cite{CT99}, which still requires $\Omega(n^{2})$ time.

The algorithm by \cite{CT99} for both Problems \ref{enu:prob-shattering} and \ref{enu:prob-shredder} has remained state-of-the-art for more than 20 years. Moreover, its approach inherently requires quadratic time because it makes $\Omega(k^2+n)$ max flow calls. Unfortunately, its quadratic runtime has become a bottleneck for the related tasks mentioned above. Naturally, one may wonder: do subquadratic algorithms exist for both problems?

\paragraph{Our Contribution.}
We answer the above question in the affirmative by showing a randomized algorithm for listing all $k$-shredders and computing a most shattering min-cut in near-linear time for all $k = \O(\polylog(n))$. Our main results are stated below.

\begin{restatable}{thm}{thmaks} \label{thm:aks}
    Let $G = (V, E)$ be an $n$-vertex $m$-edge undirected graph with vertex connectivity $k$. There exists an algorithm that takes $G$ as input and correctly lists all $k$-shredders of $G$ with probability $1 - n^{-97}$ in $\O(m + k^5 n \log^4 n)$ time.
\end{restatable}

\begin{restatable}{thm}{corshatter} \label{cor:max_sep}
   Let $G = (V, E)$ be an $n$-vertex $m$-edge undirected graph with vertex connectivity $k$. There exists a randomized algorithm that takes $G$ as input and returns a most shattering minimum vertex cut (if one exists) with probability $1 - n^{-97}$. The algorithm runs in $\O(m + k^5n \log^4 n)$ time.
\end{restatable}

\paragraph{Technical Highlights.}
Given recent developments in fast algorithms for computing vertex connectivity, one might expect that some of these modern techniques (e.g. local flow \cite{NSY19,FNSYY20}, sketching \cite{LNPSY21}, expander decomposition \cite{SY22}) will be useful for listing shredders and finding most shattering min-cuts. It turns out that they are indeed useful, but \emph{not enough}.

We have extended the techniques developed for local flow algorithms \cite{FNSYY20} to list $k$-shredders and compute the number of components that they separate. Specifically, our local algorithm lists $k$-shredders that separate the graph in an unbalanced way in time \emph{proportional to the smaller side} of the cut. To this end, we generalize the structural results related to shredders from \cite{CT99} to the local setting. To carry out this approach, we bring a new tool into the area; our algorithm queries a pairwise connectivity oracle subject to vertex failures \cite{LS22,Kos23}. Surprisingly, this is the first application of using connectivity oracles subject to vertex failures to speed up a static graph algorithm.

\paragraph{Future Work.}
Cheriyan and Thurimella also presented a dynamic algorithm for maintaining all $k$-shredders of a graph over a sequence of edge insertions/deletions. The dynamic algorithm has a preprocessing step which runs in $\O((k + \log n)kn^2 + k^3n^{3/2})$ time. It supports edge updates in $\O(m + (\min(k, \sqrt{n}) + \log n) kn)$ time and queries for a most shattering minimum vertex cut in $\O(1)$. Using our algorithm, their preprocessing time is immediately improved by a polynomial factor. We also anticipate that our algorithm could guide the design of a faster vertex connectivity augmentation algorithm.

\paragraph{Organization.}
Our paper is organized as follows. \Cref{sec:technical_overview} describes the bottleneck of the algorithm presented in \cite{CT99} and our high-level casework approach to resolving this bottleneck. We also solve the simpler case and explain why the remaining case requires more sophisticated methods. \Cref{sec:prelims} provides the preliminaries needed throughout the paper. \Cref{sec:review} explains the core algorithm in \cite{CT99} and restates important definitions and terminology. \Cref{sec:local,sec:unverified,sec:unbalanced} are devoted to solving the hard case of this problem introduced in \Cref{sec:technical_overview}. \Cref{sec:all-k-shredders} summarizes the previous sections and shows the algorithm for listing all $k$-shredders. \Cref{sec:most-shattering} provides the algorithm for finding the most shattering min-cut. Lastly, \Cref{sec:conclusion} discusses the overall results and implications for open problems.


\section{Technical Overview} \label{sec:technical_overview}
Let $G = (V, E)$ be an $n$-vertex $m$-edge undirected graph with vertex connectivity $k$. Cheriyan and Thurimella developed a deterministic algorithm called $\AKS(\cdot)$ that takes $G$ as input and lists all $k$-shredders of $G$ in $\O(knm + k^2\sqrt{n}m)$ time. They improved this bound by using the sparsification routine developed in \cite{NI92} as a preprocessing step. Specifically, there exists an algorithm that takes $G$ as input and produces an edge subgraph $G'$ on $\O(kn)$ edges such that all $k$-shredders of $G$ are $k$-shredders of $G'$ and vice versa. The algorithm runs in $\O(m)$ time. Using this preprocessing step, they obtained the bound for listing all $k$-shredders in $\O(m) + \O(k(kn)n + k^2 \sqrt{n}(kn)) = \O(k^2n^2 + k^3n^{3/2})$ time.

In this paper, we resolve a bottleneck of $\AKS(\cdot)$, improving the time complexity of listing all $k$-shredders from $\O(knm + k^2\sqrt{n}m)$ to $\O(k^4 m \log^4 n)$. Using the same sparsification routine, our algorithm runs in $\O(m + k^5 n \log^4 n)$ time.

\subsection{The Bottleneck} \label{sec:bottleneck}
A key subroutine of $\AKS(\cdot)$ is a subroutine called $\SHR(\cdot, \cdot)$ that takes a pair of vertices $(x, y)$ as input and lists all $k$-shredders that separate $x$ and $y$. This subroutine takes $\O(m)$ time plus the time to compute a flow of size $k$, which is at most $\O(km)$ time. The idea behind $\AKS(\cdot)$ is to call $\SHR(\cdot, \cdot)$ multiple times to list all $k$-shredders. $\AKS(\cdot)$ works as follows. Let $Y$ be an arbitrary set of $k + 1$ vertices. Any $k$-shredder $S$ will either separate a pair of vertices in $Y$, or separate a component containing  $Y \sm S$ from the rest of the graph.

For the first case, they call $\SHR(u, v)$ for all pairs of vertices $(u, v)\in Y \times Y$. This takes $\O(k^2)$ calls to $\SHR(\cdot, \cdot)$, which in total runs in $\O(k^3m)$ time. For the second case, we can add a dummy vertex $z$ and connect $z$ to all vertices in $Y$. Notice that any $k$-shredder $S$ separating a component containing $Y \sm S$ from the rest of the graph must separate $z$ and a vertex $v \in V \sm (Y \cup S)$. We can find these $k$-shredders by calling $\SHR(z, v)$ for all $v \in V \sm Y$.

\paragraph{The Expensive Step.}
The bottleneck of Cheriyan's algorithm is listing $k$-shredders that fall into the second case. They call $\SHR(\cdot, \cdot)$ $\Omega(n)$ times as $|V \sm Y| = n - k = \Theta(n)$ for small values of $k$. This step is precisely where the quadratic $\O(knm)$ factor comes from. To bypass this bottleneck, we categorize the problem into different cases and employ randomization.

\subsection{Quantifying a Notion of Balance} \label{sec:balanced}
Let $S$ be a $k$-shredder of $G$. A useful observation is that if we obtain vertices $x$ and $y$ in different components of $G \sm S$, then $\SHR(x, y)$ will list $S$. To build on this observation, one can reason about the relative \emph{sizes} of the components of $G \sm S$ and employ a random sampling approach. Let $C$ denote the ``largest'' component of $G \sm S$. If $C$ does not greatly outsize the remaining components, we can obtain two vertices in different components of $G \sm S$ without too much difficulty. To capture the idea of relative \emph{size} between components of $G \sm S$, we define a quantity called \emph{volume}. This definition is used commonly throughout the literature, but we have slightly altered it here.

\begin{definition}[Volume] \label{def:volume}
    For a vertex set $Q$, we refer to the \emph{volume} of $Q$, denoted by $\vol(Q)$, to denote the quantity
        $$\vol(Q) = |\{(u, v) \in E \mid u \in Q\}|.$$
    That is, $\vol(Q)$ denotes the number of edges $(u, v)$ such that either $u \in Q$ or $v \in Q$. If $\mc{Q}$ is a collection of disjoint vertex sets, then we define the volume of $\mc{Q}$ as the quantity
        $$\vol(\mc{Q}) = \vol\left(\bigcup_{Q \in \mc{Q}} Q\right).$$
\end{definition}

For convenience, we will say that a $k$-shredder $S$ admits \emph{partition} $(C, \R)$ to signify that $C$ is the largest component of $G \sm S$ by volume and $\R$ is the set of remaining components. We will also say $x \in \R$ to denote a vertex in a component of $\R$. We can categorize a $k$-shredder $S$ based on the volume of $\R$.

\begin{restatable}[Balanced/Unbalanced $k$-Shredders]{definition}{defbalanced} \label{def:balanced}
   Let $S$ be a $k$-shredder of $G$ with partition $(C, \R)$. We say $S$ is \emph{balanced} if $ \vol(\R) \geq m/k$. Conversely, we say $S$ is \emph{unbalanced} if $ \vol(\R) < m/k$.
\end{restatable}

Suppose that $S$ is a $k$-shredder with partition $(C, \R)$. Building on our discussion above, if $S$ is balanced, then $C$ does not greatly outsize the rest of the graph (by a factor of $k$ at most). Conversely, if $S$ is unbalanced, then $C$ greatly outsizes the rest of the graph.

\subsection{Listing Balanced \fmt{$k$}--Shredders via Edge Sampling} \label{subsec:balanced}
Listing all balanced $k$-shredders turns out to be straightforward via random edge sampling. Let $S$ be a $k$-shredder that admits partition $(C, \R)$. Suppose that $S$ is balanced. The idea is to sample two edges $(x,x')$ and $(y, y')$ such that $x$ and $y$ are in different components of $G \sm S$. Then, $\SHR(x, y)$ will list $S$. Intuitively, this approach works because we know that each component of $G \sm S$ cannot be too large, as $S$ is a balanced $k$-shredder. 
More precisely, by sampling $\O(k \log n)$ pairs of edges, with high probability, at least one pair of edges will hit different connected components because the biggest component $C$ has volume at most $m - m/k$. The formal statement is given below. Since its proof is standard, we defer it to \Cref{sec:omitted-proofs:2}.

\begin{restatable}{lem}{lembalanced} \label{lem:balanced}
    Let $G = (V,E)$ be an $n$-vertex $m$-edge undirected graph with vertex connectivity $k$. There exists a randomized algorithm that takes $G$ as input and returns a list $\L$ that satisfies the following. If $S$ is a balanced $k$-shredder of $G$, then $S \in \L$ with probability $1 - n^{-100}$. Additionally, every set in $\L$ is a $k$-shredder of $G$. The algorithm runs in $\O(k^2 m \log n)$ time.
\end{restatable}

\subsection{Listing Unbalanced \fmt{$k$}--Shredders via Local Flow}
Unbalanced $k$-shredders are more difficult to list than balanced $k$-shredders. For balanced $k$-shredders $S$ with partition $(C, \R)$, we can obtain two vertices in different components of $G \sm S$ without much difficulty. This is no longer the case for unbalanced $k$-shredders because $\vol(\R)$ may be arbitrarily small. If we sample two vertices, most sampled vertices will end up in $C$. Furthermore, the probability of hitting two distinct components is no longer boostable within a polylogarithmic number of sampling rounds. What this dilemma implies is that we must spend at least a linear amount of time just to collect samples. More critically, we must spend a sublinear amount of time processing an \emph{individual} sample to make any meaningful improvement. This time constraint rules out the possibility of calling $\SHR(\cdot, \cdot)$ per sample. Instead, we must develop a \emph{localized} version of $\SHR(\cdot, \cdot)$ that spends time relative to a parameter of our choice, instead of a global value such as $m$ or $n$.

In detail, consider an unbalanced $k$-shredder $S$ with partition $(C, \R)$. Observe that there must exist a power of two $2^i$ such that $2^{i-1} < \vol(\R) \leq 2^i$. If we sample $\tO(m/2^i)$ edges $(x, y)$, we will sample a vertex $x \in \R$ with high probability due to the classic hitting set lemmas. The goal is to spend only $\tO(\poly(k) \cdot 2^i)$ time processing each sample to list $S$. Suppose that such a local algorithm exists. Although we do not know the exact power of two $2^i$, we know that $\vol(\R) < \frac{m}{k}$. Hence, we can simply try all powers of two up to $\frac{m}{k}$. This gives us the near-linear runtime bound:
\begin{align*}
    \sum_{i=0}^{\lceil \log \frac{m}{k} \rceil} \tO\left(\frac{m}{2^i}\right) \cdot \tO(\poly(k) \cdot 2^i) & = \sum_{i=0}^{\lceil \log \frac{m}{k} \rceil} \tO(\poly(k) \cdot m) \\
    & = \tO(\poly(k) \cdot m).
\end{align*}
Thus, to handle unbalanced $k$-shredders, we introduce a new structural definition.

\begin{restatable}[Capture]{definition}{defcapture} \label{def:capture}
    Let $S$ be an unbalanced $k$-shredder of a graph $G = (V,E)$ with partition $(C, \R)$. Consider an arbitrary tuple $(x, \nu, \Pi)$, where $x$ is a vertex in $V$, $\nu$ is a positive integer, and $\Pi$ is a set of paths. We say that the tuple $(x, \nu, \Pi)$ \emph{captures} $S$ if the following holds.
    \begin{enumerate}
        \item $x$ is in a component of $\R$.
        \item $\frac{1}{2}\nu < \vol(\R) \leq \nu$.
        \item $\Pi$ is a set of $k$ openly-disjoint simple paths, each starting from $x$ and ending at a vertex in $C$, such that the sum of lengths over all paths is at most $k^2 \nu$.
    \end{enumerate}
\end{restatable}

At a high level, we will spend some time constructing random tuples $(x, \nu, \Pi)$ in the hopes that one of the tuples captures $S$. The main result is stated below.

\begin{restatable}{lem}{lemlocal} \label{lem:local}
    Let $G$ be a graph with vertex connectivity $k$.
    Let $(x, \nu, \Pi)$ be a tuple where $x$ is a vertex, $\nu$ is a positive integer, and $\Pi$ is a set of paths. There exists a deterministic algorithm that takes $(x, \nu, \Pi)$ as input and outputs a list $\L$ of $k$-shredders and one set $U$ such that the following holds. If $S$ is a $k$-shredder that is captured by $(x, \nu, \Pi)$, then $S \in \L$ or $S = U$. The algorithm runs in $\O(k^2 \nu \log \nu)$ time.
\end{restatable}

What is most important about this result is that we have constructed an algorithm that can identify $k$-shredders on a \emph{local} scale, i.e., proportional to an input volume parameter $\nu$. The idea behind \Cref{lem:local} is to modify $\SHR(\cdot, \cdot)$ using recent advancements in local flow algorithms.

A major caveat concerns the set $U$. We can think of $U$ as an \emph{unverified set}. Essentially, because the algorithm in \Cref{lem:local} is time-limited by a volume parameter $\nu$, we may discover a set $U$ but not have enough time to verify whether $U$ is a $k$-shredder.

\subsection{Verification via Pairwise Connectivity Oracles}
In our algorithm, we will obtain a list of $k$-shredders and a list of unverified sets. The union of these two sets will include all $k$-shredders of $G$ with high probability. However, within the list of unverified sets, there may be some false $k$-shredders. To filter out the false $k$-shredders, we utilize a pairwise connectivity oracle subject to vertex failures developed by Kosinas in \cite{Kos23}. The idea is that to determine whether an unverified set $S$ is a $k$-shredder, we will make some connectivity queries between vertex pairs $(u, v)$ to determine whether $u$ is connected to $v$ in $G \sm S$. These local queries and additional structural observations will help us determine whether $G \sm S$ contains at least three components.

\subsection{Summary}
We have combined the $\SHR(\cdot, \cdot)$ subroutine with local flow algorithms in \cite{FNSYY20} to construct an algorithm that lists $k$-shredders on a local scale. However, this local algorithm potentially lists an unverified set that may not be a $k$-shredder. To filter out unverified sets that are not $k$-shredders, we used a pairwise connectivity oracle subject to vertex failures developed in \cite{Kos23} and new structural insights.

%% file: src/3,4-prelim,CT99.tex
\section{Preliminaries} \label{sec:prelims}
This paper concerns finite, undirected, and unweighted graphs with vertex connectivity $k$. We use standard graph-theoretic definitions. Let $G = (V,E)$ be a graph with vertex connectivity $k$. For a vertex subset $S \subseteq V$, we use $G \sm S$ to denote the subgraph of $G$ induced by removing all vertices in $S$. A \emph{connected component} (component for short) of a graph refers to any maximally-connected subgraph, or the vertex set of such a subgraph. Suppose that $S$ is a $k$-shredder of $G$. The \emph{largest component} of $G \sm S$ is the component with the greatest \emph{volume}, where volume is defined in \Cref{def:volume}. We break ties arbitrarily. For convenience, we will say that a $k$-shredder $S$ admits partition $(C, \R)$ to signify that $C$ is the largest component of $G \sm S$ and $\R$ is the set of remaining components. We will also write $x \in \R$ to indicate a vertex in a component of $\R$.

For a vertex subset $Q \subseteq V$, we define the set of neighbors of $Q$ as the set $N(Q) = \{v \in V \sm Q \mid (u, v) \in E , u \in Q \}$. For vertex subsets $Q_1, Q_2 \subseteq V$, we use $E(Q_1, Q_2)$ to denote the set of edges with one endpoint in $Q_1$ and one endpoint in $Q_2$. Specifically, $E(Q_1, Q_2) = \{(u, v) \in E \mid u \in Q_1, v \in Q_2\}$.

Let $\pi$ be a path in $G$.  We refer to the \emph{length} of a path $\pi$ as the number of edges in $\pi$. Say $\pi$ starts at a vertex $x \in V$. We say the \emph{far-most endpoint} of $\pi$ to denote the other endpoint of $\pi$. Although all paths we refer to are undirected, our usage of the far-most endpoint will be unambiguous. Let $x$ and $y$ be two vertices used by $\pi$. We denote $\pi[x \e y]$ as the sub-path of $\pi$ from $x$ to $y$. We denote $\pi(x \e y)$ as the sub-path $\pi[x \e y]$ \emph{excluding} the vertices $x$ and $y$. We say two paths $\pi_1, \pi_2$ are \emph{openly-disjoint} if they share no vertices except their endpoints. Let $\Pi$ be a set of paths. We say that $\Pi$ is openly disjoint if all pairs of paths in $\Pi$ are openly disjoint. We say $v \in \Pi$ to denote a vertex used by one of the paths in $\Pi$ and $(u, v) \in \Pi$ to denote an edge used by one of the paths in $\Pi$.

\section{Cheriyan and Thurimella's Algorithm} \label{sec:review}
We will use $\SHR(\cdot, \cdot)$ as a subroutine and extend it to a \emph{localized setting}. To do this, we need to review the terminology and ideas presented in \cite{CT99}.

\begin{restatable}[\cite{CT99}]{thm}{thmpairwise} \label{thm:pairwise-shredders}
    Let $G$ be an $n$-vertex $m$-edge undirected graph with vertex connectivity $k$. Let $x$ and $y$ be two distinct vertices. There exists a deterministic algorithm $\SHR(\cdot, \cdot)$ that takes $(x, y)$ as input and returns all $k$-shredders of $G$ separating $x$ and $y$ in $\O(km)$ time.
\end{restatable}

At a high level, $\SHR(\cdot, \cdot)$ works as follows. Let $x$ and $y$ be two vertices in a $k$-vertex-connected graph $G$. Firstly, we use a flow algorithm to obtain a set $\Pi$ of $k$ openly-disjoint simple paths from $x$ to $y$. Observe that any $k$-shredder $S$ of $G$ separating $x$ and $y$ must contain exactly one vertex from each path of $\Pi$. More crucially, at least one component of $G \sm S$ must also be a component of $G \sm \Pi$. That is, for some component $Q$ of $G \sm \Pi$, we have $N(Q) = S$. This is the main property that $\SHR(\cdot, \cdot)$ exploits. It lists potential $k$-shredders by finding components $Q$ of $G \sm \Pi$ such that $|N(Q)| = k$ and $N(Q)$ consists of exactly one vertex from each path in $\Pi$ (see \Cref{fig:shredders-xy}).
\begin{figure}[!ht]
    \centering
    \includegraphics[scale=0.4]{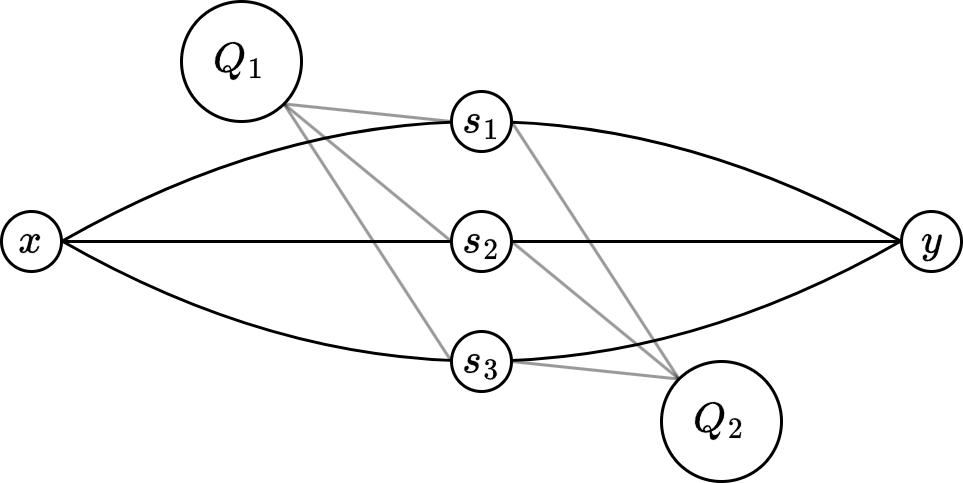}
    \caption{Let $\Pi$ be the set of three openly-disjoint simple paths from $x$ to $y$. A call to $\SHR(x, y)$ will identify $N(Q_1) = N(Q_2)$ as potential $3$-shredders.}
    \label{fig:shredders-xy}
\end{figure}
We can state these facts formally with the following definitions.

\begin{definition}[Bridge]
    A \emph{bridge} $\Gamma$ of $\Pi$ is a component of $G \sm \Pi$ or an edge $(u, v) \in E$ such that $(u, v) \notin \Pi$, but $u \in \Pi$ and $v \in \Pi$.
\end{definition}

If $\Gamma$ is an edge, we define $\vol(\Gamma) = 1$. Otherwise, $\vol(\Gamma)$ is defined according to \Cref{def:volume}. For a set of bridges $\B$, we define $\vol(\B) = \sum_{\Gamma \in \B} \vol(\Gamma)$.

\begin{definition}[Attachments] \label{def:attachments}
     Let $\Gamma$ be a bridge of $\Pi$. If $\Gamma$ is a component of $G \sm \Pi$, then the set of \emph{attachments} of $\Gamma$ is the vertex set $N(\Gamma)$. Otherwise, if $\Gamma$ is an edge $(u, v) \in E$, the set of \emph{attachments} of $\Gamma$ is the vertex set $\{u, v\}$.
\end{definition}

For convenience, we refer to a single vertex among the attachments of $\Gamma$ as an attachment of $\Gamma$. For a path $\pi \in \Pi$ and bridge $\Gamma$ of $\Pi$, we denote $\pi(\Gamma)$ as the set of attachments of $\Gamma$ that are in $\pi$. If $Q$ is an arbitrary vertex set, then we use $\pi(Q)$ to denote the set $Q \cap \pi$. If $\pi(Q)$ is a singleton set, then we may use $\pi(Q)$ to represent the unique vertex in $Q \cap \pi$. In all contexts, it will be clear what object the notation is referring to.

\begin{definition}[$k$-Tuple]
    A set $S$ is called a $k$-tuple with respect to $\Pi$ if $|S| = k$ and for all $\pi \in \Pi$, we have $|S \cap \pi| = 1$.
\end{definition}

We adopt the following notation as in $\SHR(\cdot, \cdot)$. If a component $\Gamma$ of $G \sm \Pi$ is such that $N(\Gamma)$ forms a $k$-tuple with respect to $\Pi$, then 
the set $S = N(\Gamma)$ is called a \emph{candidate $k$-shredder}, or \emph{candidate} for short. In general, not all candidates are true $k$-shredders. To handle this, $\SHR(\cdot, \cdot)$ performs a \emph{pruning} phase by identifying fundamental characteristics of false candidates. To identify false candidates, we formalize some of the key definitions in \cite{CT99}.

\begin{definition}[$\delta_{\pi}$]
    Let $\pi$ be a path starting from a vertex $x$. For a vertex $v \in \pi$, we define $\delta_{\pi}(v)$ as the distance between $x$ and $v$ along path $\pi$.
\end{definition}

\begin{definition}[$\preceq, \succeq$] \label{def:order}
    Let $\Gamma_1, \Gamma_2$ be two bridges of $\Pi$. We say that $\Gamma_1 \preceq \Gamma_2$ if for all paths $\pi \in \Pi$, for all pairs $(u, v) \in \pi(\Gamma_1) \times \pi(\Gamma_2)$ we have $\delta_\pi(u) \leq \delta_\pi(v)$. We use an identical definition for $\Gamma_1 \succeq \Gamma_2$ by replacing $\delta_\pi(u) \leq \delta_\pi(v)$ with $\delta_{\pi}(u) \geq \delta_{\pi}(v)$.
\end{definition}

\begin{definition}[Straddle] \label{def:straddle}
     Let $\Gamma_1, \Gamma_2$ be two bridges of $\Pi$. We say $\Gamma_1$ \emph{straddles} $\Gamma_2$ (or $\Gamma_2$ straddles $\Gamma_1$) if there exist paths (not necessarily distinct) $\pi_1, \pi_2 \in \Pi$ such that there exists a vertex pair $(u_1, u_2) \in \pi_1(\Gamma_1) \times  \pi_1(\Gamma_2)$ satisfying $\delta_{\pi_1}(u_1) < \delta_{\pi_1}(u_2)$ and $(v_1, v_2) \in \pi_2(\Gamma_1) \times \pi_2(\Gamma_2)$ satisfying $\delta_{\pi_2}(v_1) > \delta_{\pi_2}(v_2)$.
\end{definition}

For a visual description of \Cref{def:straddle}, see figure \Cref{fig:def-straddle}. It will also be useful to notice that the statement $\Gamma_1$ straddles $\Gamma_2$ is equivalent to $\Gamma_1 \not \preceq \Gamma_2$ and $\Gamma_1 \not \succeq \Gamma_2$. Furthermore, \Cref{def:order,def:straddle} are still well-defined even for candidates (i.e. we can use the $\preceq, \succeq$ operators to compare bridges to candidates and candidates to candidates).

\begin{figure}[!ht]
    \centering
    \includegraphics[scale=0.5]{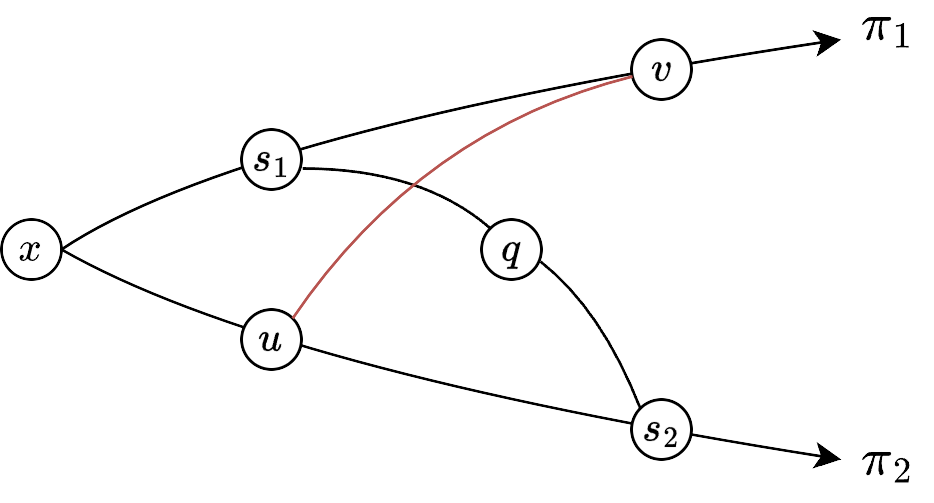}
    \caption{Here we have $\Pi = \{\pi_1,  \pi_2\}$. The candidate $\{s_1, s_2\}$ is straddled by the edge $(u, v)$ because $\delta_{\pi_2}(u) < \delta_{\pi_2}(s_2)$ and $\delta_{\pi_1}(v) > \delta_{\pi_1}(s_1)$.}
    \label{fig:def-straddle}
\end{figure}

\paragraph{Candidate Pruning.}
The key is that a candidate $S$ reported by $\SHR(x, y)$ is a $k$-shredder of $G$ if and only if no bridge of $\Pi$ straddles $S$. Hence, to prune false candidates, $\SHR(x, y)$ finds all bridges of $\Pi$ that are straddled using standard techniques such as radix sort and interval merging. These notions are formalized in the following two lemmas.

\begin{restatable}{lem}{lemxystraddle} \label{lem:xy-straddle}
    Let $x$ and $y$ be two distinct vertices and let $\Pi$ denote a set of $k$ openly disjoint simple paths from $x$ to $y$. Let $S$ be a candidate with respect to $\Pi$. Then, $S$ is a $k$-shredder separating $x$ and $y$ if and only if no bridge of $\Pi$ straddles $S$.
\end{restatable}

\begin{restatable}{lem}{lemprune} \label{lem:prune}
    Let $x$ be a vertex in $G$ and let $\Pi$ denote a set of $k$ openly-disjoint simple paths starting from $x$. Let $\ell$ denote the sum of the lengths over all paths in $\Pi$. Let $\C, \B$ be a set of candidates of $\Pi$ and a set of bridges of $\Pi$, respectively. There exists a deterministic algorithm that, given $(\C, \B, \Pi)$ as input, lists all candidates $S \in \C$ such that $S$ is not straddled by another candidate in $\C$ nor by any bridge in $\B$. The algorithm runs in $\O(k |\C| \log |\C| + \ell + \vol(\B))$ time.
\end{restatable}

Note that \Cref{lem:prune} generalizes $\SHR(\cdot, \cdot)$ in the sense that it does not require that all paths in $\Pi$ were between a pair of vertices $(x, y)$. This will be particularly useful in the later sections. Because the proofs of \Cref{lem:xy-straddle,lem:prune} directly follow from the ideas from \cite{CT99}, we defer their proofs along with the proof of \Cref{thm:pairwise-shredders} to \Cref{sec:omitted-proofs:4}.

%% file: src/6-local.tex
\section{Locally Listing \fmt{$k$-}-Shredders or an Unverified Set} \label{sec:local}
Suppose we have obtained a tuple $(x, \nu, \Pi)$ that \emph{captures} a $k$-shredder $S$. Recall \Cref{def:capture} for the definition of capture. The method for constructing $(x, \nu, \Pi)$ that captures $S$ is via random sampling and is deferred to \Cref{sec:unbalanced}. In this section, we show a local algorithm that, assuming $S$ is captured by $(x, \nu, \Pi)$, puts $S$ in the returned list $\cal{L}$ or reports $S$ as an unverified set $U$. How we handle the unverified set will be explained in \Cref{sec:unverified}.

\lemlocal*

To prove \Cref{lem:local}, we first state some structural properties of shredders when captured by $(x,\nu,\Pi)$ in \Cref{sec:local structure}. Then, we describe the algorithm of in \Cref{sec:local algorithm} and analyze it in \Cref{sec:local analysis}.

\subsection{Structures of Local Shredders} \label{sec:local structure}
The following lemmas translate the structural properties of shredders from \Cref{sec:review} to the setting when they are captured by the tuple $(x,\nu,\Pi)$. They will be needed in the analysis of our algorithm.

\begin{fact} \label{fact:attachment-shredder}
     Let $S$ be a $k$-shredder with partition $(C, \R)$ captured by $(x, \nu, \Pi)$. For every path $\pi \in \Pi$, we have $| S \cap \pi | = 1$.
\end{fact}

\begin{proof}
    Observe that all paths in $\Pi$ must cross $S$ because $x \notin C$ and each path starts from $x$ and ends in $C$. Moreover, each path in $\Pi$ must cross exactly one unique vertex of $S$ because there are only $k$ vertices in $S$ and the paths are openly-disjoint. Hence, it is impossible for all paths to be openly-disjoint if one path uses more than one vertex of $S$.
\end{proof}

The following statement is the local analog to the statement discussed below \Cref{thm:pairwise-shredders}.

\begin{lem} \label{lem:preserve-components}
    Let $S$ be a $k$-shredder with partition $(C, \R)$ captured by $(x, \nu, \Pi)$. Let $Q_x$ denote the component in $\R$ containing $x$. Every component in $\R \sm \{Q_x\}$ is a component of $G \sm \Pi$.
\end{lem}

\begin{proof}
    By \Cref{fact:attachment-shredder}, $\Pi$ cannot contain a vertex in any component of $\R \sm \{Q_x\}$, as such a path would necessarily use more than one vertex of $S$. Thus, for an arbitrary component $Q \in \R \sm \{Q_x\}$, we have $\Pi \cap Q = \varnothing$. It follows that $Q \subseteq Q'$ for some component $Q'$ of $G \sm \Pi$.
    
    We show that $Q' \subseteq Q$. To see this, suppose that $Q' \sm Q \neq \varnothing$. Notice $N(Q') = S$, so $\Gamma$ must contain a vertex of $S$. However $S \subseteq \Pi$, which contradicts the fact that $Q'$ is a component of $G \sm \Pi$. Therefore, we have $Q \subseteq Q'$ and $Q' \subseteq Q$, which implies $Q = Q'$.
\end{proof}

Next, we prove one direction of the ``if and only if'' statement of  \Cref{lem:xy-straddle} in the local setting. Another direction requires knowing how our algorithm constructs candidate $k$-shredders and will be proved later in \Cref{lem:local:no-straddle-is-shredder}.

\begin{lem} \label{lem:shredders-not-straddled}
    Let $S$ be a $k$-shredder with partition $(C, \R)$ captured by $(x, \nu, \Pi)$. Then, no bridge of $\Pi$ straddles $S$. 
\end{lem}

\begin{proof} 
    Suppose that there exists a bridge $\Gamma$ of $\Pi$ that straddles $S$. Let $z$ be a vertex in $C$. We will construct a path from $x$ to $z$ in $G \sm S$. This will contradict $S$ admitting partition $(C, \R)$ because $x \in \R$ and $z \in C$.
    
    Since $\Gamma$ straddles $S$, there exist paths $\pi_1, \pi_2$ in $\Pi$ and vertices $u \in \pi_1(\Gamma), v \in \pi_2(\Gamma)$ such that $\delta_{\pi_1}(u) < \delta_{\pi_1}(\pi_1(S))$ and $\delta_{\pi_2}(v) > \delta_{\pi_2}(\pi_2(S))$. First, we show there exists a path $P$ from $u$ to $v$ in $G \sm S$. If $\Gamma$ is the edge $(u, v)$, the path $P = \{u, v\}$ exists in $G \sm S$ as $u \notin S$ and $v \notin S$. Otherwise, $\Gamma$ is a component $Q$ of $G \sm \Pi$ and both $u$ and $v$ are in $N(Q)$. This means there exist edges $(u, u')$ and $(v, v')$ in $G \sm S$ such that $u'$ and $v'$ are in $Q$. Because $Q$ is a component of $G \sm \Pi$, $v$ and $v'$ must be connected in $G \sm S$, as $S \subseteq \Pi$. Therefore, $u$ and $v$ are connected in $G \sm S$ via the vertices $u'$ and $v'$.
    
    Using $P$ as a subpath, we can construct a $x \e z$ path in $G \sm S$ as follows. Consider the path obtained by joining the intervals: $\pi_1[x \e u] \cup P \cup \pi_2[v \e z]$. Note that $\pi_1[x \e u]$ exists in $G \sm S$ because $\delta_{\pi_1}(u) < \delta_{\pi_1}(\pi_1(S))$, so no vertex of $S$ lives in the interval $\pi_1[x \e u]$. Similarly, $\pi_2[v \e z]$ exists in $G \sm S$ because $\delta_{\pi_2}(v) > \delta_{\pi_2}(\pi_2(S))$. Finally, as $P$ exists in $G \sm S$, the joined path from $x$ to $z$ exists in $G \sm S$.
\end{proof}

\subsection{Algorithm and Runtime}
\label{sec:local algorithm}
\paragraph{Outline.}
At a high level, \Cref{alg:local} works as follows. Let $(x, \nu, \Pi)$ be a tuple given as input to the algorithm. For each path $\pi \in \Pi$, we can traverse the path from $x$ to the far-most endpoint of $\pi$. At a given vertex $u \in \pi$, we can explore the bridges of $\Pi$ attached to $u$ using a breadth-first search (BFS). While doing so, we maintain a list of bridges of $\Pi$. We also maintain a list of candidates by checking whether the attachment set of a bridge forms a $k$-tuple. Among the list of candidates, false candidates are then pruned correctly and efficiently via \Cref{lem:shredders-not-straddled} and \Cref{lem:prune}. So far, we have not deviated from $\SHR(\cdot, \cdot)$ (see \Cref{fig:capture}).
\begin{figure}[!ht]
    \centering
    \includegraphics[scale=0.22]{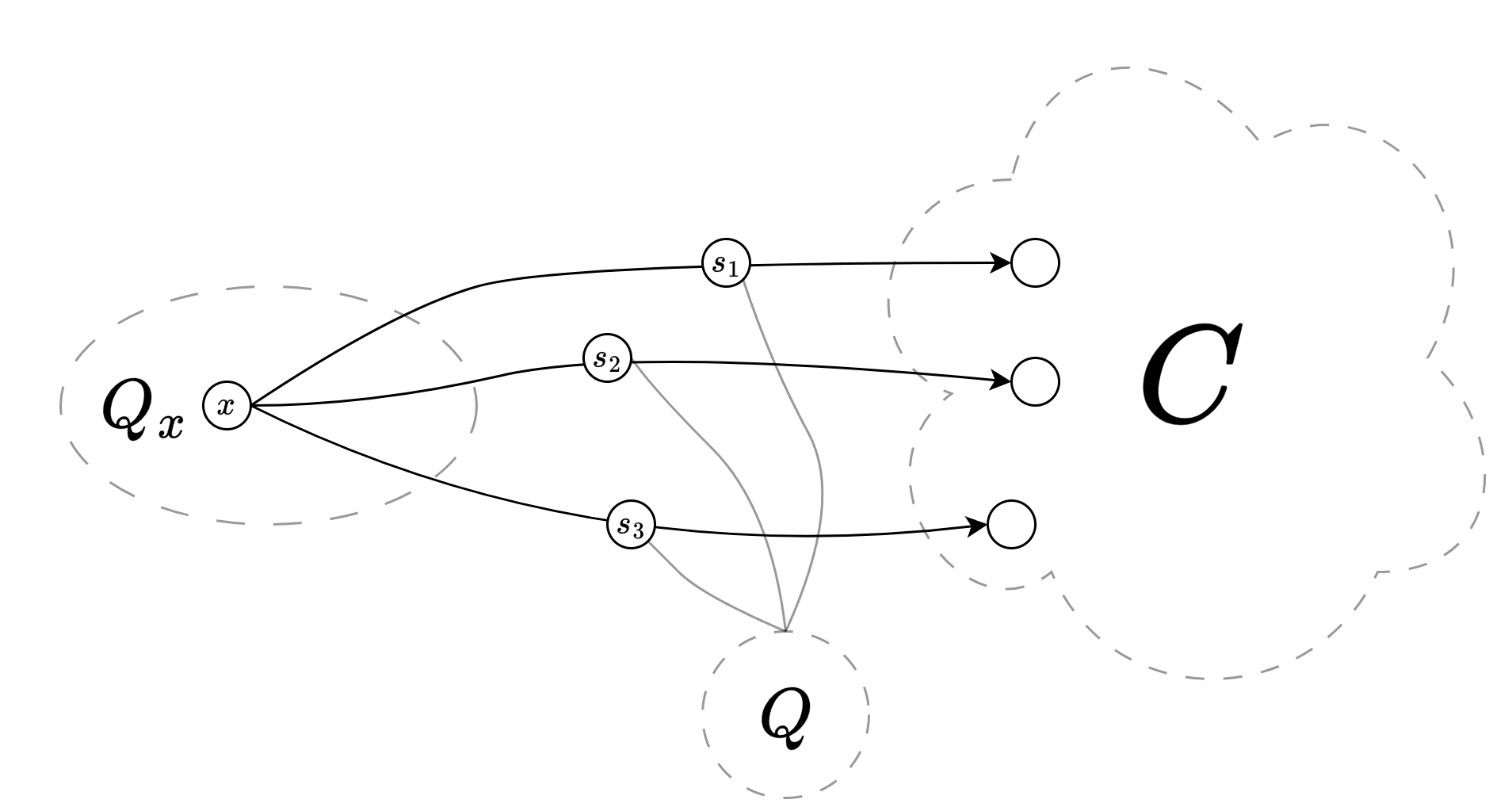}
    \caption{Suppose the $k$-shredder $S = \{s_1, s_2, s_3\}$ is captured by $(x, \nu, \Pi)$. We can recover the component $Q$ by exploring bridges of $\Pi$ using BFS.}
    \label{fig:capture}
\end{figure}
However, our modification imposes a restriction on the number of edges we are allowed to explore via BFS. We can explore at most $\nu$ edges along each path. Suppose we are processing a vertex $u$ along path $\pi$. If the number of edges explored for this path exceeds $\nu$ while exploring bridges of $\Pi$ attached to $u$, we terminate early and mark $u$ as an \textit{unverified vertex}. Intuitively, the unverified vertex $u$ serves as a boundary of exploration. It signifies that bridges attached to $u$ were not fully explored, so we flag $u$ and treat it with extra caution during a later step. At the end of processing all the paths in $\Pi$, we will have obtained a list of bridges of $\Pi$, a list of candidate $k$-shredders of $\Pi$, and an \textit{unverified set} consisting of all the unverified vertices. The final step is to black box \Cref{lem:prune} to prune false candidates among the list of candidate $k$-shredders. Consider \Cref{alg:local}.

\begin{algorithm}[!ht]
\caption{Finding the unbalanced $k$-shredders with small components containing $x$}
\label{alg:local}
\begin{algorithmic}[1]
    \Require{$x$ - a sample vertex \\
        $\nu$ - volume parameter \\
        $\Pi$ - a set of $k$ openly disjoint paths starting from $x$
    }
    \Ensure{$\L$ - a list of $k$-shredders captured by $(x, \nu, \Pi)$ \\
        $U$ - an unverified set
    }

    \State $\L \leftarrow \varnothing$
    \State $U \leftarrow \varnothing$
    
    \For {each path $\pi \in \Pi$} \label{alg:local:all-paths}
        \State reset explored edges and vertices to null
        \For {each vertex $u \in \pi$ in increasing distance from $x$} \label{alg:local:explore-path}
            \If {$u$ is the far-most endpoint of $\pi$} \label{alg:local:endpoint}
                \State $U \leftarrow U \cup \{u\}$
                \State \textbf{break}
            \EndIf
            
            \State explore all unexplored bridges of $\Pi$ attached to $u$ with BFS: \label{alg:local:bfs}
            \INDSTATE terminate early the moment more than $\nu$ edges are explored \label{alg:local:explore-limit}

            \If {BFS terminated early}
                \State $U \leftarrow U \cup \{u\}$ \label{alg:local:term-early-unverified}
                \State \textbf{break}
            \Else \label{alg:local:not-term-early}
                \For{each bridge $\Gamma$ explored by BFS} \label{alg:local:process-bridges}
                    \If {the attachment set $A$ of $\Gamma$ forms a $k$-tuple of $\Pi$}
                        \State $\L \leftarrow \L \cup \{A\}$ \label{alg:local:mark-candidate}
                    \EndIf
                \EndFor
            \EndIf
        \EndFor
    \EndFor

    \label{alg:local:initial-candidates}
    \If {$x \in U$} \label{alg:local:x-in-U}
        \State \textbf{return} $(\varnothing, \varnothing)$
    \EndIf

    \item[]
    \State prune false $k$-shredders in $\L \cup \{U\}$ using \Cref{lem:prune} \label{alg:local:prune}
    \State $F \leftarrow$ the set of far-most endpoints of all paths in $\Pi$ 
    \If {$U$ was not pruned \textbf{and} $U \cap F = \varnothing$} \label{alg:local:U-prune}
        \State \textbf{return} $(\L, U)$
    \Else
        \State \textbf{return} $(\L, \varnothing)$
    \EndIf
\end{algorithmic}
\end{algorithm}

\paragraph{Small Remarks.}
A few lines of \Cref{alg:local} were added to handle some special cases. Specifically, lines \ref{alg:local:endpoint}, \ref{alg:local:x-in-U}, and \ref{alg:local:U-prune}. The justification of these lines will become clear in the following pages, but we will provide basic reasoning here.

The first thing to remember is that unlike $\SHR(\cdot, \cdot)$, the paths in $\Pi$ may have different far-most endpoints. The danger arises from the following case. Suppose there is a component $Q$ of $G \sm \Pi$ such that $N(Q)$ is equal to the set of endpoints of the paths in $\Pi$. We should not report $N(Q)$ as a candidate $k$-shredder because there does not necessarily exist a vertex $y \notin Q$ such that $x$ and $y$ are disconnected in $G \sm N(Q)$. Contrast this with $\SHR(\cdot, \cdot)$, where all the $k$ paths in $\Pi$ are between two vertices $x$ and $y$. In $\SHR(\cdot, \cdot)$, we know that any component $Q$ in $G \sm \Pi$ and its attachments set $N(Q)$ cannot contain $y$. If $N(Q)$ is not straddled by any bridge, then there are three components of $G \sm N(Q)$: the component containing $x$, the component containing $y$, and $Q$.

In our case, however, each path starts from $x$ but can end at an arbitrary vertex. There is no common endpoint $y$. The fix for this is the inclusion of line \ref{alg:local:endpoint}. If we reach an endpoint of a path, we should immediately mark the vertex as unverified and consider the endpoints later. Intuitively, because $(x, \nu, \Pi)$ captures $S$, we do not need to explore the far-most endpoints of $\Pi$ because they live in $C$. After all, our goal is to list $S$, which does not contain any of the path endpoints.

Line \ref{alg:local:x-in-U} covers a trivial case. What may occur is that the bridges of $\Pi$ attached to $x$ have total volume greater than $\nu$. In this case, we will terminate the BFS due to line \ref{alg:local:explore-limit} when processing $x$ and never find a candidate $k$-shredder. Right before line \ref{alg:local:x-in-U} we will have $U = \{x\}$. In this case, we can terminate immediately and return empty sets because no candidate $k$-shredders were found. We should think of this case as a ``bad input'' tuple scenario where $(x, \nu, \Pi)$ did not capture anything.

Line \ref{alg:local:U-prune} covers the case where the unverified set contains an endpoint of some path in $\Pi$. We need not report $U$ in this case because any $k$-shredder $S$ that $(x, \nu, \Pi)$ captures must use vertices strictly preceding each far-most endpoint of all paths. This may seem like an arbitrary restriction, but it will be useful to include line \ref{alg:local:U-prune} in the proof of \Cref{lem:unverified-set}.


\paragraph{Running Time.}
Before showing the correctness of this algorithm in the next subsection, we analyze the running time here.

\begin{lem} \label{lem:local:time}
    Algorithm \ref{alg:local} on input $(x, \nu, \Pi)$ runs in $\O(k^2 \nu \log \nu)$ time.
\end{lem}

\begin{proof}
    For each path, we explore a total of $\O(\nu)$ volume. This is because we keep track of exploration and terminate early whenever we explore more than $\nu$ edges. Because there are $k$ paths, the entire \textbf{for} loop starting on line \ref{alg:local:all-paths} can be implemented using hash sets in $\O(k \nu)$ time.

    Now we analyze the rest of the algorithm. As argued above, the total number of edges we explore over all paths is at most $k \nu$. Since each candidate $k$-shredder corresponds with a distinct bridge of $\Pi$ whose attachments form a $k$-tuple, there can be at most $\frac{k \nu}{k} = \O(\nu)$ candidate $k$-shredders in $\L$. Let $\B$ denote the set of all bridges explored. Any bridge of $\Pi$ at the very least requires one edge, so the volume of all bridges put together is at most $\O(k \nu)$. Since the sum of lengths over all paths is at most $k^2 \nu$, \Cref{lem:prune} implies we can implement line \ref{alg:local:prune} in $\O(k |\L| \log |\L| + k^2\nu + \vol(\B)) = \O(k\nu \log \nu + k^2 \nu + k \nu) = \O(k^2 \nu \log \nu)$ time.
\end{proof}

\subsection{Correctness} \label{sec:local analysis}
Now we prove the correctness portion of \Cref{lem:local}. These lemmas effectively streamline the argument of \cite{CT99} and show that the structural properties of $\SHR(\cdot, \cdot)$ extend to our localized setting. We will prove two core lemmas.

\begin{restatable}{lem}{lemlocalidentify} \label{lem:local:identify}
    Let $S$ be a $k$-shredder captured by $(x, \nu, \Pi)$. Then, \Cref{alg:local} will either identify $S$ as a candidate $k$-shredder or identify $S$ as the unverified set. Specifically, either $S \in \L$ or $S = U$ by the end of the algorithm.
\end{restatable}

\begin{restatable}{lem}{lemlocalvalidity} \label{lem:local:validity}
    All sets in $\L$ at the end of \Cref{alg:local} are $k$-shredders.
\end{restatable}

It is straightforward to check that these two lemmas directly imply the algorithm correctness portion of \Cref{lem:local}.

\subsubsection{Identification of Captured Shredders}
First we prove \Cref{lem:local:identify}. Let $S$ be a $k$-shredder captured by $(x, \nu, \Pi)$. We will prove that the algorithm will walk along each path and reach the vertices of $S$. Then, we prove that $S$ will either be marked as a candidate $k$-shredder or the unverified set.

\begin{lem} \label{lem:reach}
    Let $S$ be a $k$-shredder captured by $(x, \nu, \Pi)$ with partition $(C, \R)$. Then, $S \preceq U$.
\end{lem}

\begin{proof}
    Fix a path $\pi \in \Pi$. Let $P = \{v \in \pi \mid \delta_{\pi}(v) < \delta_{\pi}(\pi(S))\}$. It suffices to prove that the algorithm will not terminate early when exploring bridges attached to $P$. Let $Q_x$ denote the component of $\R$ containing $x$. Every vertex in $P$ must be in $Q_x$ because they precede $\pi(S)$ along $\pi$. Because all vertices of $S$ are in $\Pi$, we cannot reach a component of $\R \sm \{Q_x\}$ when exploring bridges of $\Pi$ attached to $P$. It follows that all edges explored for the vertices in $P$ must contain an endpoint in $Q_x$. Because $Q_x \in \R$ we have $\vol(Q_x) \leq \vol(\R) \leq \nu$. Line \ref{alg:local:explore-limit} ensures the algorithm will not terminate early. This means that for all paths $\pi \in \Pi$, the unverified vertex $u \in \Pi$ satisfies $\delta_{\pi}(\pi(S)) \leq \delta_{\pi}(u)$. Hence, $S \preceq U$.
\end{proof}

\begin{proof}[Proof of \Cref{lem:local:identify}]
    \Cref{lem:preserve-components} implies that there exists a component $Q$ of $G \sm S$ that is also a component of $G \sm \Pi$. This means that $Q$ is a bridge of $\Pi$ whose attachment set is $S$. From \Cref{lem:reach} we know that there are two cases. Either $\delta_\pi(\pi(S)) = \delta_\pi(\pi(U))$ for all paths $\pi \in \Pi$, or $\delta_\pi(\pi(S)) < \delta_\pi(\pi(U))$ for some path $\pi \in \Pi$. In the former case, we have that $S = U$. In the latter case, line \ref{alg:local:bfs} and lines \ref{alg:local:process-bridges}-\ref{alg:local:mark-candidate} imply we must have fully explored all bridges of $\Pi$ attached to a vertex of $S$. Because $Q$ is a bridge of $\Pi$ attached to all vertices of $S$, we will put $N(Q) = S$ as a set in $\L$. \Cref{lem:shredders-not-straddled} implies that no bridge of $\Pi$ straddles $S$, which means $S$ will not be pruned from $\L$ on line \ref{alg:local:prune}.
\end{proof}

\subsubsection{Algorithm Validity}
So far we have proved that all $k$-shredders captured by $(x, \nu, \Pi)$ will be returned by \Cref{alg:local}. However, this is not enough to show correctness; we also need to show that every set returned by the algorithm is a $k$-shredder. To do this, we will first prove that if a set $S \in \L$ is not straddled by a bridge of $\Pi$, then it is a $k$-shredder. Then, we will prove that straddled sets in $\L$ are pruned on line \ref{alg:local:prune}. The following fact is trivial but useful for the next two lemmas.

\begin{fact} \label{fact:local:candidate-precedes-U}
    Let $S$ be a set remaining in $\L$ after line \ref{alg:local:prune}. Then, there exists a path $\pi \in \Pi$ such that $\delta_{\pi}(\pi(S)) < \delta_{\pi}(u)$, where $u$ is the unverified vertex along $\pi$.
\end{fact}

\begin{proof}
    Since $S$ was put in $\L$, there must exist a vertex $s \in S$ such that the algorithm explored all bridges of $\Pi$ attached to $s$. This is guaranteed by lines \ref{alg:local:term-early-unverified} and \ref{alg:local:not-term-early}. This implies that $s$ precedes the unverified vertex along its path.
\end{proof}

Next we prove that candidates that are not straddled by any bridges of $\Pi$ are $k$-shredders. This can be viewed as a local version of one direction of \Cref{lem:xy-straddle}.

\begin{lem} \label{lem:local:no-straddle-is-shredder}
    Let $S$ be a set put in $\L$ at some point on line \ref{alg:local:mark-candidate}. If no bridge of $\Pi$ straddles $S$, then $S$ is a $k$-shredder.
\end{lem}

\begin{proof}
    We prove $S$ is a $k$-shredder by showing $G \sm S$ has at least three components. Firstly, we have $S = N(Q)$ for some component $Q$ of $G \sm \Pi$. Fix a path $\pi \in \Pi$ such that the vertex $\pi(S)$ is not the far-most endpoint of $\pi$. Such a path must exist by \Cref{fact:local:candidate-precedes-U}. Let $z$ denote the far-most endpoint of $\pi$. It suffices to show that $x$ is not connected to $z$ in $G \sm S$. Then, because $x \notin Q$ and $z \notin Q$, there are at least three components of $G \sm S$: the component containing $x$, the component containing $z$, and $Q$.

    Suppose for the sake of contradiction that there exists a simple path $P$ from $x$ to $z$ in $G \sm S$. Our plan is to show that $S$ is straddled by a bridge of $\Pi$, contradicting the assumption that no bridge straddles $S$. To construct such a bridge, let $u$ denote the last vertex along $P$ such that $u \in \pi_1$ for some path $\pi_1 \in \Pi$ and $\delta_{\pi_1}(u) < \delta_{\pi_1}(\pi_1(S))$. Such a vertex must exist because $x \in P$ and satisfies $\delta_{\pi}(x) < \delta_{\pi}(\pi(S))$. Similarly, let $v$ denote the first vertex along $P$ after $u$ such that $v \in \pi_2$ for some path $\pi_2 \in \Pi$ and $\delta_{\pi_2}(v) > \delta_{\pi_2}(\pi_2(S))$. Such a vertex must exist because $z \in P$ and satisfies $\delta_{\pi}(z) > \delta_{\pi}(\pi(S))$.

    Consider the open subpath $P(u \e v)$ (see \Cref{fig:subpath-bridge}). We claim that every vertex in $P(u \e v)$ is not in $\Pi$. To see this, suppose that there exists a vertex $w \in P(u \e v)$ such that $w \in \pi_3$ for some path $\pi_3 \in \Pi$. Then, $\delta_{\pi_3}(w) = \delta_{\pi_3}(\pi_3(S))$ because otherwise, it would contradict the definition of $u$ or $v$. This implies $w = \pi_3(S)$ because the vertex distance $\delta_{\pi_3}(\pi_3(S))$ from $x$ along $\pi_3$ is unique. However, this contradicts the definition of $P$ because $P$ is a path in $G \sm S$: it cannot contain $\pi_3(S)$.
    
    \begin{figure}[!ht]
        \centering
        \includegraphics[scale=0.3]{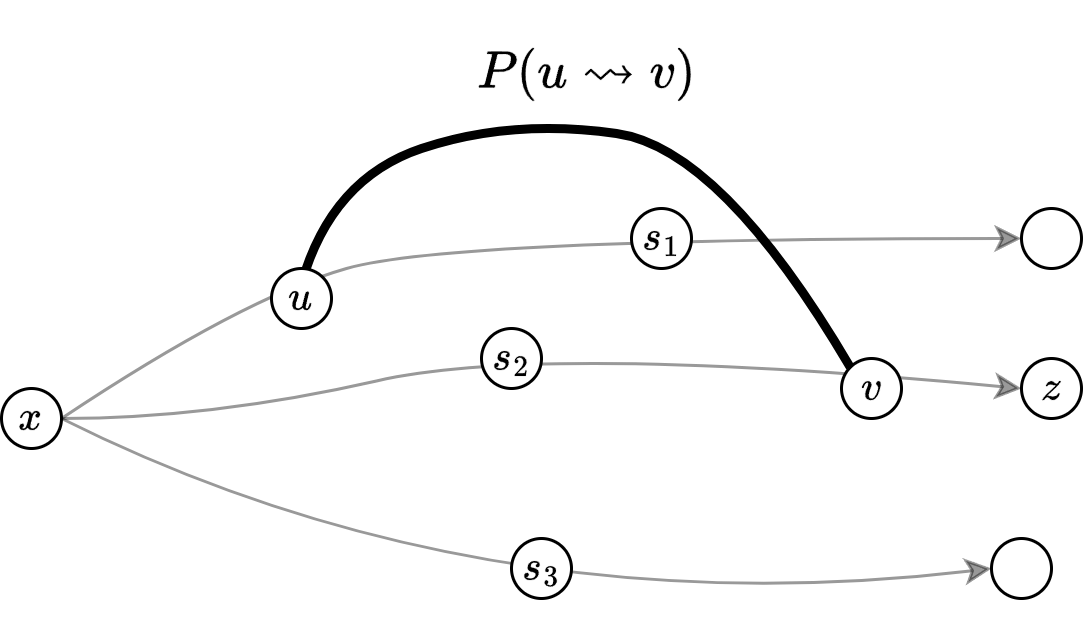}
        \caption{The existence of an $x \e z$ path implies the existence of vertices $u$ and $v$ (it is possible that $x = u$ and $v = z$), which form the attachments of a straddling bridge.}
        \label{fig:subpath-bridge}
    \end{figure}
    
    Hence, the inner vertices in the open subpath $P(u \e v)$ exist as a bridge $\Gamma$ of $\Pi$ with attachment set $\{u, v, \ldots \}$. If there are no vertices in the subpath, then $\Gamma$ is the edge $(u, v)$. More importantly, we have that $\delta_{\pi_1}(u) < \delta_{\pi_1}(\pi_1(S))$ and $\delta_{\pi_2}(v) > \delta_{\pi_2}(\pi_2(S))$. Therefore, $\Gamma$ straddles $S$, arriving at the desired contradiction.
\end{proof}

\begin{lem} \label{lem:local:prune-straddled}
    Let $S$ be a candidate $k$-shredder detected by \Cref{alg:local} on input $(x, \nu, \Pi)$. If there exists a bridge of $\Pi$ that straddles $S$, then $S$ will be pruned from the list of candidates. Similarly, if there exists a bridge that straddles the unverified set $U$, then $U$ will be pruned.
\end{lem}

\begin{proof}
    Let $S$ be the candidate straddled by a bridge $\Gamma$ of $\Pi$. For a visual example of this proof, see \Cref{fig:detect-straddle}. Let $A$ denote the attachment set of $\Gamma$. If any vertex in $A$ has been fully processed without terminating early, then $\Gamma$ will be detected on line \ref{alg:local:bfs}. In this case, $S$ will be pruned by line \ref{alg:local:prune}.

    Hence, we can assume that no vertex of $A$ has been fully explored, i.e. $U \preceq A$. We will show that $U$ straddles $S$, in which case $S$ will be pruned by line \ref{alg:local:prune}. First, by \Cref{fact:local:candidate-precedes-U}, there must exist a path $\pi_1 \in \Pi$ that satisfies $\delta_{\pi_1}(\pi_1(S)) < \delta_{\pi_1}(\pi_1(U))$. Because $\Gamma$ straddles $S$, there must exist a path $\pi_2 \in \Pi$ and a vertex $v \in A \cap \pi_2$ such that $\delta_{\pi_2}(v) < \delta_{\pi_2}(\pi_2(S))$. Furthermore, because $U \preceq A$, this implies that $\delta_{\pi_2}(\pi_2(U)) \leq \delta_{\pi_2}(v) < \delta_{\pi_2}(\pi_2(S))$. Now, we have $\delta_{\pi_1}(\pi_1(S)) < \delta_{\pi_1}(\pi_1(U))$ and $\delta_{\pi_2}(\pi_2(U)) < \delta_{\pi_2}(\pi_2(S))$, which implies $U$ straddles $S$. Line \ref{alg:local:prune} implies that $S$ will be pruned from the list of candidates.
    
    \begin{figure}[!ht]
        \centering
        \includegraphics[scale=0.26]{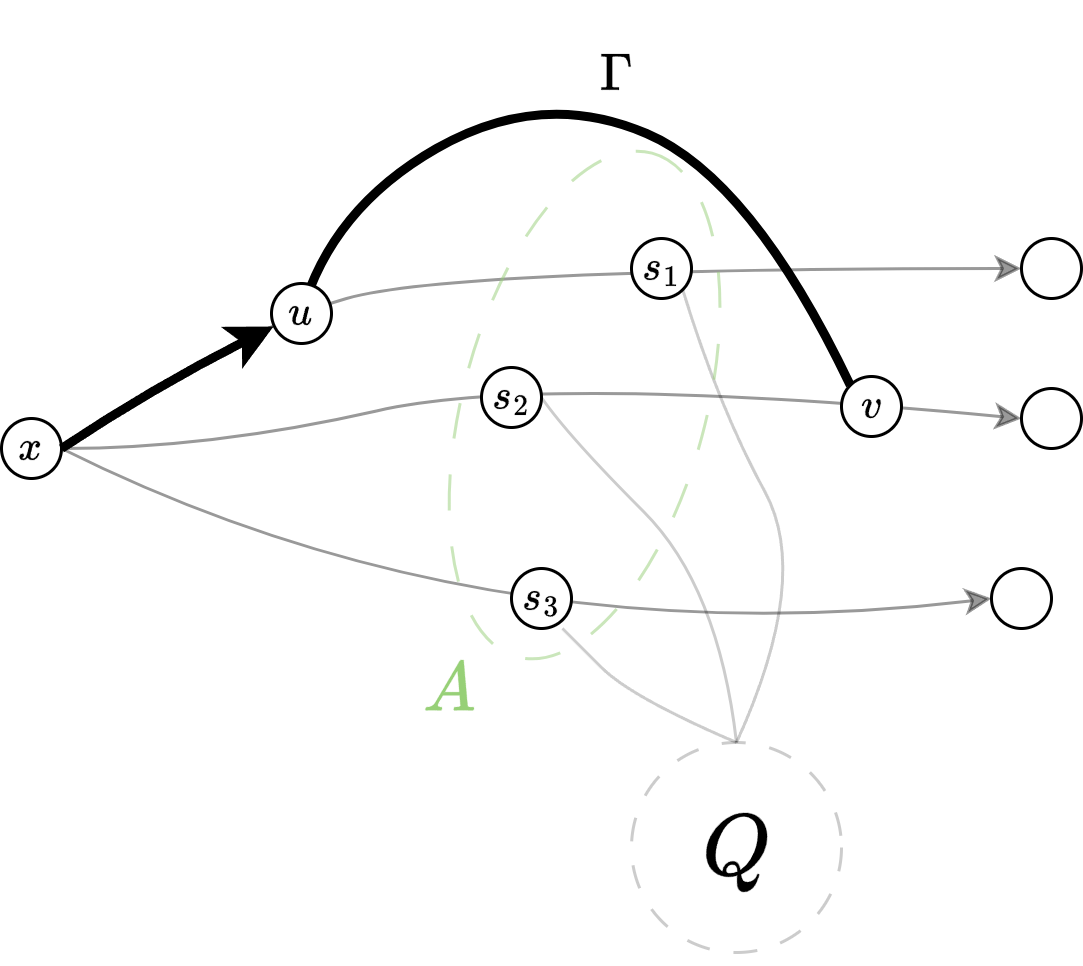}
        \includegraphics[scale=0.26]{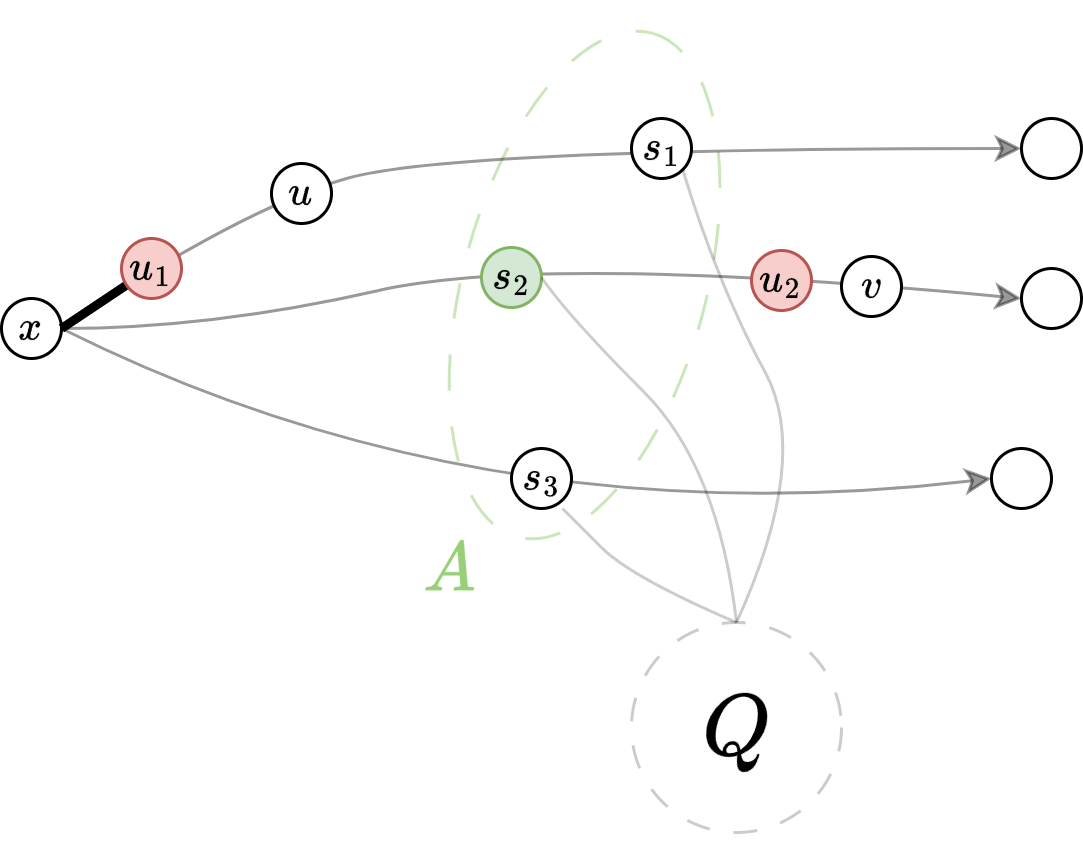}
        \caption{The bridge $Q$ has attachment set $A = S = \{s_1, s_2, s_3\}$. $S$ is straddled by $\Gamma$ with attachments $\{u, v\}$. If we explore all bridges of $\Pi$ attached to $u$, then $\Gamma$ will be naturally detected (top). Otherwise, we must have terminated early at some vertex $u_1$ before reaching $u$. \Cref{fact:local:candidate-precedes-U} implies all bridges of $\Pi$ attached to some vertex $s_2$ were explored without terminating early. So $s_2$ must precede some unverified vertex $u_2$. This implies that $S$ is straddled by the unverified vertices (bottom).}
        \label{fig:detect-straddle}
    \end{figure}

    Next, we prove the second statement. Suppose the unverified set $U$ is straddled by a bridge.  First, we may assume that $|U| = k$ as otherwise $U = \{x\}$, which is caught by line \ref{alg:local:x-in-U}. Suppose that a bridge $\Gamma$ of $\Pi$ straddles $U$ with attachment set $A$. If any vertex in $A$ has been fully explored by the algorithm, then $\Gamma$ will be detected by the BFS and $U$ will be pruned by line \ref{alg:local:prune}. Otherwise, we have $U \preceq A$, which contradicts our assumption that $\Gamma$ straddles $U$.
\end{proof}

\begin{proof}[Proof of \Cref{lem:local:validity}]
    \Cref{lem:local:no-straddle-is-shredder} shows that the sets in $\L$ which are not straddled by a bridge of $\Pi$ are $k$-shredders. \Cref{lem:local:prune-straddled} shows that sets in $\L$ (along with the unverified set) which \emph{are} straddled are removed from $\L$. This implies that at the end of \Cref{alg:local}, every set in $\L$ is a $k$-shredder.
\end{proof}

Finally, we can combine \Cref{lem:local:identify,lem:local:validity,lem:local:time} to prove \Cref{lem:local}.

\begin{proof}[Proof of \Cref{lem:local}]
    Let $S$ be a $k$-shredder captured by $(x, \nu, \Pi)$. \Cref{lem:local:identify} implies that $S \in \L$ or $S = U$ by the end of the algorithm. Furthermore, every set in $\L$ by the end of the algorithm is a $k$-shredder by \Cref{lem:local:validity}. Finally, \Cref{lem:local:time} proves the time complexity of the algorithm.
\end{proof}

\subsection{Additional Lemmas}
We provide two more lemmas that will be useful in the later sections. We defer their proofs to the appendix, as they are morally repeating the arguments used above.

\begin{restatable}{lem}{lemunverifiedset} \label{lem:unverified-set}
    Suppose that \Cref{alg:local} on input $(x, \nu, \Pi)$ returns an unverified set $U$ such that $U \neq \varnothing$. Then, $U$ is a $k$-separator. Specifically, let $z$ be the far-most endpoint of an arbitrary path $\pi \in \Pi$. Then, $x$ and $z$ are not connected in $G \sm U$.
\end{restatable}

\begin{restatable}{lem}{lemtrackvolume} \label{lem:track-local-volume}
    Suppose that \Cref{alg:local} returns a nonempty unverified set $U$ on input $(x, \nu, \Pi)$. Let $Q_x$ denote the component of $G \sm U$ containing $x$. There exists a modified version of \Cref{alg:local} that also computes $\vol(Q_x)$. The modification requires $\O(k\nu)$ additional time, subsumed by the runtime of \Cref{alg:local}.
\end{restatable}

%% file: src/7-unverified.tex
\section{Resolving Unverified Sets} \label{sec:unverified}
We now address the case where \Cref{alg:local} returns an unverified set $U$ on input $(x, \nu, \Pi)$. The difference between $U$ and the list of returned $k$-shredders $\L$ is that the algorithm did not find a component $Q$ of $G \sm \Pi$ such that $N(Q) = U$. In some sense, this means that we ``lose'' a component of $G \sm U$. The danger in guessing that $U$ is a $k$-shredder is that $U$ might only be a $k$-separator. Imagine that $G \sm U$ has exactly two components: $Q_1$ and $Q_2$. If $\vol(Q_2) \gg \vol(Q_1)$, then it becomes quite challenging to determine whether there exists a third component of $G \sm U$ in $\O(\vol(Q_1))$ time. To cope with this uncertainty, we make another structural classification.

\begin{definition}[Low-Degree, High-Degree] \label{def:low-high-deg}
    Let $S$ be a $k$-shredder with partition $(C, \R)$. Let $\nu$ be the unique power of two satisfying $\frac{1}{2} \nu < \vol(\R) \leq \nu$. We say that $S$ has \textit{low-degree} if there exists a vertex $s \in S$ such that $\deg(s) \leq \nu$. Otherwise, we say $S$ has \textit{high-degree.}
\end{definition}

One aspect of this definition may seem strange: we have specifically described $\nu$ as a power of two. This is because in the later sections, we will use geometric sampling to capture $k$-shredders. The sampling parameters we use will be powers of two. Hence, we have imposed this slightly arbitrary detail on \Cref{def:low-high-deg}. For now, all that matters is that $\frac{1}{2} \nu < \vol(\R) \leq \nu$.

\paragraph{Organization.}
There are two main lemmas for this section. \Cref{lem:low-deg-aux} handles low-degree unverified sets and is presented in \Cref{sec:low-deg}. \Cref{lem:high-deg} handles high-degree unverified sets and is presented in \Cref{sec:high-deg}. The idea is that after \Cref{alg:local} returns an unverified set $U$, we will use \Cref{lem:low-deg-aux} to test whether $U$ is a low-degree $k$-shredder. If not, we will leave keep $U$ as a potential high-degree $k$-shredder. After all unverified sets have been returned, we use \Cref{alg:high-deg} to extract all high-degree $k$-shredders from the remaining unverified sets.

\subsection{Low-Degree} \label{sec:low-deg}
Suppose that $U$ is a low-degree $k$-shredder with partition $(C, \R)$. Suppose that the tuple $(x, \nu, \Pi)$ captures $U$ and \Cref{alg:local} reports $U$ as an unverified set. Our goal is to design an algorithm that confirms $U$ is a $k$-shredder in the same time complexity as \Cref{alg:local} up to $\polylog(n)$ and $\poly(k)$ factors.

\Cref{lem:unverified-set} gives us two vertices in different components of $G \sm U$: $x$ and $z$ (where $z$ is the far-most endpoint of any path in $\Pi$). Since $U$ is a $k$-shredder, there must exist a third component of $G \sm U$ and every vertex in $U$ must be adjacent to a vertex in this third component. Because $U$ is low-degree, there must exist a vertex $u \in U$ with $\deg(u) \leq \nu$. The idea is to scan through the edges adjacent to this low-degree vertex $u$ and find an edge $(u, y)$ such that $y$ is neither connected to $x$ nor $z$ in $G \sm U$. To make the scanning procedure viable, we need to efficiently answer pairwise connectivity queries in $G \sm U$. The key ingredient is a result obtained by Kosinas in \cite{Kos23}. Specifically, the following holds.

\begin{thm}[Connectivity Oracle under Vertex Failures] \label{thm:pairwise-conn-oracle}
    There exists a deterministic data structure for an undirected graph $G$ on $n$ vertices and $m$ edges with $\O(km \log n)$ preprocessing time that supports the following operations.
    \begin{enumerate}
        \item Given a set $F$ of $k$ vertices, perform a data structure update in time $\O(k^4 \log n)$.
        \item Given a pair of vertices $(x, y)$ return $\TRU$ if $x$ is connected to $y$ in $G \sm F$ in time $\O(k)$.
    \end{enumerate}
\end{thm}

We defer the method of capturing $k$-shredders for later. For now, let us assume that the tuple $(x, \nu, \Pi)$ captures a $k$-shredder $U$ and $U$ is reported as an unverified set by \Cref{alg:local}. We show an auxiliary algorithm that will be helpful in listing low-degree $k$-shredders.

\begin{lem} \label{lem:low-deg-aux}
    After preprocessing the input graph using the data structure from \Cref{thm:pairwise-conn-oracle}, there exists a deterministic algorithm that takes in as input $(x, \nu, \Pi, U)$, where $U$ is an unverified set returned by \Cref{alg:local} on input $(x, \nu, \Pi)$. The algorithm outputs $\TRU$ if $U$ is a $k$-shredder and there exists a vertex $u \in U$ such that $\deg(u) \leq \nu$ in $\O(k^4 \log n + k\nu)$ time.
\end{lem}

\begin{remark} \label{rem:LS-oracle}
    In \cite{LS22}, Long and Saranurak showed the existence of a pairwise connectivity oracle subject to vertex failures with $\O(k^2n^{o(1)})$ update time and $\O(k)$ query time. Our usage of Kosinas's oracle involves making one update operation with a set of size $k$ and $\O(\nu)$ many connectivity queries. Using Long and Saranurak's version, we can improve the dependency on $k$ in the time complexity of \Cref{alg:low-deg} to $\O(k^2n^{o(1)} + k\nu)$ time. 
\end{remark}

\paragraph{Algorithm Outline.}
As mentioned above, the idea is to scan through the edges adjacent to $u$ in order to find a vertex $y$ such that $y$ is neither connected to $x$ nor $z$ in $G \sm U$ (see \Cref{fig:low-deg}). We will be utilizing \Cref{thm:pairwise-conn-oracle} to determine whether a pair of vertices $(x, y)$ are connected in $G \sm U$. Pseudocode is given in \Cref{alg:low-deg}.

\begin{figure}[!ht]
    \centering
    \includegraphics[scale=0.3]{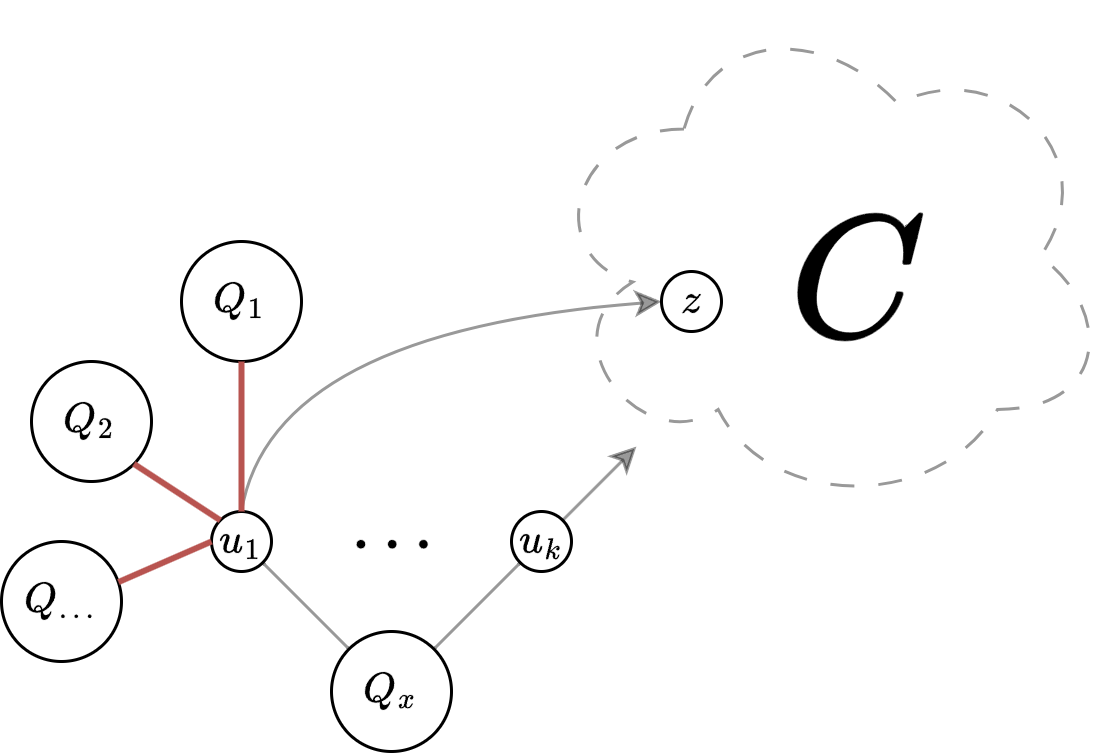}
    \caption{Here the set $U = \{u_1, \ldots, u_k\}$ was reported as an unverified set. Suppose there exists a vertex (without loss of generality) $u_1$ such that $\deg(u_1) \leq \nu$. Then, we can simply scan through the neighbors of $u_1$ in $\O(\nu)$ iterations to obtain an edge incident to a component of $G \sm U$ that is not $Q_x$ nor $C$.}
    \label{fig:low-deg}
\end{figure}

\begin{algorithm}[!ht]
\caption{Auxiliary algorithm for low-degree $k$-shredders}
\label{alg:low-deg}
\begin{algorithmic}[1]
    \Require{$x$ - a sample vertex \\
        $\nu$ - volume parameter \\
        $\Pi$ - a set of $k$ openly-disjoint paths starting from $x$ \\
        $U$ - unverified set returned by \Cref{alg:local} on $(x, \nu, \Pi)$
    }
    \Ensure{$\TRU$ if (1) and (2) are satisfied, $\FLS$ otherwise \\
        \phantom{ZZ} (1) $U$ is a $k$-shredder \\
        \phantom{ZZ} (2) there exists a vertex $u \in U$ such that $\deg(u) \leq \nu$
    }
    \If {all vertices in $U$ have degree greater than $\nu$} \label{alg:low-deg:all-high-deg}
        \State \textbf{return} $\FLS$
    \EndIf
    \State $u \leftarrow$ a vertex in $U$ with $\deg(u) \leq \nu$ \label{alg:low-deg:u}
    \State $\pi \leftarrow$ the path in $\Pi$ containing $u$ \label{alg:low-deg:pi}
    \State $z \leftarrow$ the far-most endpoint of $\pi$ \label{alg:low-deg:z}

    \item[]
    \State $\Phi \leftarrow$ the pairwise connectivity oracle from \Cref{thm:pairwise-conn-oracle} (already preprocessed)
    \State update $\Phi$ on vertex failure set $U$ \label{alg:low-deg:update-U}
    
    \For {each edge $(u, y)$ adjacent to $u$ such that $y \notin U$} \label{alg:low-deg:adj-u}
        \State $q_{x,y} \leftarrow$ connectivity query between $x$ and $y$ in $G \sm U$ \label{alg:low-deg:x-y}
        \State $q_{y,z} \leftarrow$ connectivity query between $y$ and $z$ in $G \sm U$ \label{alg:low-deg:y-z}
        
        \If {$q_{x,y}$ is $\FLS$ \textbf{and} $q_{y,z}$ is $\FLS$} \label{alg:low-deg:queries-false}
            \State \textbf{return} $\TRU$
        \EndIf
    \EndFor

    \State \textbf{return} $\FLS$
\end{algorithmic}
\end{algorithm}

\begin{lem} \label{lem:low-deg-aux:correctness}
    Let $(x, \nu, \Pi, U)$ be an input tuple to \Cref{alg:low-deg}. Then, the algorithm returns $\TRU$ if and only if $U$ is a $k$-shredder and there exists a vertex $u \in U$ such that $\deg(u) \leq \nu$.
\end{lem}

\begin{proof}
    Suppose that $U$ is a $k$-shredder and that there exists a vertex $u \in U$ such that $\deg(u) \leq \nu$. Let $u$ be the vertex found by line \ref{alg:low-deg:u}. Let $\pi$ be the unique path in $\Pi$ that contains $u$ and let $z$ be the far-most endpoint of $\pi$. These are found by lines \ref{alg:low-deg:pi}-\ref{alg:low-deg:z}. By \Cref{lem:unverified-set}, we have at least two components of $G \sm U$ so far: $Q_x$ (the component containing $x$) and $Q_z$ (the component containing $z$). Since $U$ is a $k$-shredder, there exists a third component $Q_y$ of $G \sm U$. Let $y$ be a vertex in $Q_y$ adjacent to $u$. Note that $x, y, z$ are each in distinct components of $G \sm U$, which implies they are pairwise disconnected in $G \sm U$. We will find $(u, y)$ by scanning through the edges adjacent to $u$. Prior to calling the algorithm, we will initialize a global pairwise connectivity oracle subject to vertex failures as in \Cref{thm:pairwise-conn-oracle}. We update the oracle with respect to $U$ on line \ref{alg:low-deg:update-U} so we can answer pairwise connectivity queries in $G \sm U$. The rest is straightforward. The \textbf{for} loop on line \ref{alg:low-deg:adj-u} iterates through all edges adjacent to $u$. Hence, the edge $(u, y)$ described above will be found. For each edge $(u, y)$, we perform the two pairwise connectivity queries in $G \sm U$ by lines \ref{alg:low-deg:x-y}-\ref{alg:low-deg:y-z}. Lastly, we will discover that $x, y, z$ are all pairwise disconnected in $G \sm U$ and confirm that $U$ is a $k$-shredder by line \ref{alg:low-deg:queries-false}.
    
    Now suppose that $U$ is not a $k$-shredder or there does not exist a vertex $u \in U$ with $\deg(u) \leq \nu$. If $U$ is not a $k$-shredder, then there are exactly two components in $G \sm U$: the component containing $x$, and the component containing $z$ (the far-most endpoint of any path in $\Pi$). Therefore, for every edge $(u, y)$ we iterate through such that $u \in U, y \notin U$, either $x$ and $y$ are connected in $G \sm U$ or $y$ and $z$ are connected in $G \sm U$. Line \ref{alg:low-deg:queries-false} guarantees that we will not return $\TRU$. If there does not exist a vertex $u \in U$ with $\deg(u) \leq \nu$, we can return $\FLS$ immediately. This is done by line \ref{alg:low-deg:all-high-deg}.
\end{proof}

\begin{lem} \label{lem:low-deg-aux:time}
    \Cref{alg:low-deg} on input $(x, \nu, \Pi, U)$ runs in $\O(k^4 \log n + k\nu)$ time.
\end{lem}

\begin{proof} 
    Line \ref{alg:low-deg:all-high-deg} can be implemented in $\O(k)$ time by iterating over all vertices in $U$ and checking the sizes of their adjacency lists. Lines \ref{alg:low-deg:u}-\ref{alg:low-deg:z} can be done in constant time. The update on line \ref{alg:low-deg:update-U} takes $\O(k^4 \log n)$ time according to \Cref{thm:pairwise-conn-oracle}. Iterating over the edges adjacent to $u$ takes $\O(\nu)$ time because $\deg(u) \leq \nu$. For each edge, we make two pairwise connectivity queries, each of which takes $\O(k)$ time according to \Cref{thm:pairwise-conn-oracle}. This gives the total time bound of $\O(k) + \O(k^4 \log n) + \O(k\nu) = \O(k^4 \log n + k \nu)$.
\end{proof}

\begin{proof}[Proof of \Cref{lem:low-deg-aux}]
    The correctness and time complexity of \Cref{alg:low-deg} are given by \Cref{lem:low-deg-aux:correctness} and \Cref{lem:low-deg-aux:time}.
\end{proof}

\subsection{High-Degree} \label{sec:high-deg}
It is quite difficult to determine locally whether $U$ is a high-degree $k$-shredder. We can no longer hope for a low-degree vertex $u \in U$ and scan through its edges to find a third component of $G \sm U$. To tackle high-degree $k$-shredders reported as unverified sets, our idea is to ignore them as they are reported and filter them later using one subroutine. Let $U$ be a high-degree $k$-shredder with partition $(C, \R)$. Let $\nu$ be the power of two satisfying $\frac{1}{2} \nu < \vol(\R) \leq \nu$. Every vertex $x \in \R$ has $\deg(x) \leq \nu$ because $\vol(\R) \leq \nu$. Furthermore, every vertex $u \in U$ has $\deg(u) > \nu$ because $U$ is high-degree. We can exploit this structure by noticing that $U$ forms a ``wall'' of high-degree vertices. That is, if we obtain a vertex $x \in \R$, we can use BFS seeded at $x$ to explore a small area of vertices with degree at most $\nu$. The high degree vertices of $U$ would prevent the graph traversal from escaping the component of $G \sm U$ containing $x$. At the end of the traversal, the set of explored vertices would compose a component $Q$ of $G \sm U$ and we can report that $N(Q) = U$ might be a high-degree $k$-shredder. To obtain a vertex $x \in \R$, we will use a classic hitting set lemma.

\begin{restatable}[Hitting Set Lemma]{lem}{lemhittingset} \label{lem:hitting-set}
    Let $Q$ be a set of vertices. Let $\nu$ be a positive integer satisfying $\frac{1}{2}\nu < \vol(Q) \leq \nu$. If we independently sample $400 \frac{m}{\nu} \log n$ edges $(u, v)$ uniformly at random, we will obtain an edge $(u, v)$ such that $u \in Q$ with probability $1 - n^{-100}$.
\end{restatable}

Now we have a tool to obtain a vertex in $\R$. We are not finished yet, as we need to connect this idea with the unverified sets reported by \Cref{alg:local}. Suppose that $U$ was reported by \Cref{alg:local} on input $(x, \nu, \Pi)$. We have that $x$ is in some component $Q$ of $G \sm U$. With \Cref{lem:hitting-set}, our goal is to sample a vertex $y$ in a \textit{different} component $Q_y$ of $G \sm U$. As described above, we can exploit the fact that $U$ is a high-degree $k$-shredder by exploring low-degree vertices in the neighborhood of $y$ to recover $Q_y$. At this point, we have essentially recovered two components of $G \sm U$: $Q$ and $Q_y$. To make further progress, we present a short, intuitive lemma.

\begin{lem} \label{lem:volume-checksum}
    Let $U$ be a $k$-separator. Let $Q_1, Q_2$ be two distinct components of $G \sm U$. Then $U$ is a $k$-shredder if and only if $\vol(Q_1) + \vol(Q_2) + |E(U, U)| < m$.
\end{lem}

\begin{proof}
    We prove the forward direction first. If $U$ is a $k$-shredder, then there exists a third component $Q_3$ of $G \sm U$. Specifically, there must exist an edge $(u, v)$ where $u \in U$ and $v \in Q_3$. Notice that $(u, v) \notin E(U, U)$ because $v \notin U$. Furthermore, $(u, v)$ is not counted in $\vol(Q_1)$ nor $\vol(Q_2)$ because $(u, v)$ is not incident to any vertex in $Q_1 \cup Q_2$. This implies that $\vol(Q_1) + \vol(Q_2) + |E(U, U)| < m$.

    For the backward direction, notice that if $\vol(Q_1) + \vol(Q_2) + |E(U, U)| = m$, we can argue that any edge $(u, v)$ falls into three categories. Either $(u, v)$ is adjacent to a vertex in $Q_1$, or $(u, v)$ is adjacent to a vertex in $Q_2$, or $u \in U$ and $v \in U$. This implies that the existence of an edge from $U$ to a third component of $G \sm U$ is impossible, which means $G \sm U$ only has two components.
\end{proof}

We are finally ready to handle high-degree $k$-shredders reported as unverified sets. We present the main result.

\begin{lem} \label{lem:high-deg}
    There exists a randomized Monte Carlo algorithm that takes as input a list $\U$ of tuples of the form $(x, \nu, \Pi, U)$, where $U$ is an unverified set returned by \Cref{alg:local} on $(x, \nu, \Pi)$. The algorithm returns a list of $k$-shredders $\L$ that satisfies the following. If $(x, \nu, \Pi, U) \in \U$ is a $k$-tuple such that $U$ is a high-degree $k$-shredder, then $U \in \L$ with probability $1 - n^{-100}$. The algorithm runs in $\O(k^2 m \log^2 n)$ time.
\end{lem}

\paragraph{Algorithm Outline.}
We summarize the argument made above. Fix a tuple $(x, \nu, \Pi, U)$ in $\U$ such that $U$ is a high-degree $k$-shredder with partition $(C, \R)$. Let $\nu$ denote the unique power of two satisfying $\frac{1}{2}\nu < \vol(\R) \leq \nu$. Because $U$ is high-degree, we have for all $u \in U$, $\deg(u) > \nu$. Let $Q_x$ denote the component of $G \sm U$ containing $x$. Suppose we have obtained a vertex $y$ in a component $Q_y \in \R \sm \{Q_x\}$. Since all vertices in $Q_y$ have degree at most $\nu$, the idea is to explore all vertices connected to $y$ that have degree at most $\nu$. Because $N(Q_y) = U$ and each vertex in $U$ has degree greater than $\nu$, we will compute $Q_y$ as the set of explored vertices. After computing $Q_y$, we will have obtained two components of $G \sm U$: $Q_x$ and $Q_y$. In order to confirm that $U$ is a $k$-shredder, all that is left to do is to make sure that $Q_x$ and $Q_y$ are not the only components of $G \sm U$. \Cref{lem:volume-checksum} implies this can be done by verifying that $\vol(Q_x) + \vol(Q_y) + |E(U, U)| < m$. Now, the final step is to find such a vertex $y$. We use a random sampling procedure for this task. Consider \Cref{alg:high-deg}. The following two lemmas prove the correctness and time complexity of \Cref{alg:high-deg}.

\begin{algorithm}[!ht]
\caption{Extracting all high-degree $k$-shredders from unverified sets}
\label{alg:high-deg}
\begin{algorithmic}[1]
    \Require{$\U$ - a list of tuples of the form $(x, \nu, \Pi, U)$}
    \Ensure{$\L$ - a list of $k$-shredders that satisfies the following: \\ \phantom{ZZ}
        if $(x, \nu, \Pi, U) \in \U$ is such that $U$ is a high-degree $k$-shredder, \\ \phantom{ZZ} 
        then $U \in \L$ with probability $1 - n^{-100}$
    }
    \State $\L \leftarrow \varnothing$
    \For {$i \leftarrow 0$ to $\lceil \log \left(\frac{m}{k}\right) \rceil$} \label{alg:high-deg:nu-loop}
        \State $\nu \leftarrow 2^i$ \label{alg:high-deg:nu}
        \For{$400 \frac{m}{\nu} \log n$ times} \label{alg:high-deg:num-samples}
            \State independently sample an edge $(y, z)$ uniformly at random \label{alg:high-deg:sample-edge}
            \State $Q_y \leftarrow$ set of vertices explored by BFS seeded at $y$: \label{alg:high-deg:BFS}
                \INDSTATE ignore vertices $v$ such that $\deg(v) > \nu$ \label{alg:high-deg:ignore-high-deg}
                \INDSTATE terminate early the moment more than $\nu$ edges are explored \label{alg:high-deg:explore-limit}
            \If {BFS terminated early}
                \State \textbf{skip} to next sample
            \ElsIf {there exists a tuple $(x, \nu_x, \Pi, U) \in \U$ such that $N(Q_y) = U$} \label{alg:high-deg:U-exists}
                \State $Q_x \leftarrow$ component of $G \sm U$ containing $x$
                \If {$x \notin Q_y$ \textbf{and} $\vol(Q_x) + \vol(Q_y) + |E(U, U)| < m$} \label{alg:high-deg:checksum}
                    \State $\L \leftarrow \L \cup \{U\}$ \label{alg:high-deg:add-U-to-L}
                \EndIf
            \EndIf
        \EndFor
    \EndFor
    \State \textbf{return} $\L$
\end{algorithmic}
\end{algorithm}

\begin{lem} \label{lem:high-deg-correctness}
    Let $\U$ be the list of tuples given as input to \Cref{alg:high-deg} and let $\L$ be the output list. Suppose $U$ is a $k$-shredder such that there exists a $k$-tuple $(x, \nu, \Pi, U) \in \U$. Then, $U \in \L$ with probability $1 - n^{-100}$. Furthermore, every set in $\L$ is a $k$-shredder.
\end{lem}

\begin{proof}
    Fix a tuple $(x, \nu, \Pi, U)$ in $\U$ such that $U$ is a high-degree $k$-shredder with partition $(C, \R)$. We will show that $U \in \L$ with high probability. Let $Q_x$ denote the component in $\R$ containing $x$. Let $Q_y$ denote a component of $\R \sm \{Q_x\}$. We will first show that the algorithm obtains a vertex $y \in Q_y$ with high probability. Let $\nu$ denote the unique power of two satisfying $\frac{1}{2} \nu < \vol(Q_y) \leq \nu$. \Cref{lem:hitting-set} implies that if we sample $400 \frac{m}{\nu} \log n$ edges, we will obtain an edge $(y, z)$ such that $y \in Q_y$. Although we do not know the explicit value of $\nu$, all powers of two up to $\frac{m}{\nu}$ are tried by the algorithm during the loop on line \ref{alg:high-deg:nu-loop}. Hence, we will eventually obtain an edge $(y, z)$ with $y \in Q_y$ with probability $1 - n^{-100}$ during the \textbf{for} loop on line \ref{alg:high-deg:num-samples}.

    After obtaining a vertex $y \in Q_y$, we perform BFS seeded at $y$. Since $\frac{1}{2} \nu < \vol(Q_y) \leq \nu$, we know that all vertices in $Q_y$ have degree at most $\nu$. We also know that $Q_y$ can be computed by traversing at most $\nu$ edges because $\vol(Q_y) \leq \nu$. Lastly, we know that all vertices in $U$ have degree greater than $\nu$ because $U$ is a high-degree $k$-shredder. This implies $Q_y$ will be computed by line \ref{alg:high-deg:BFS}. After computing $Q_y$, we can determine whether there exists a tuple $(x, \nu_x, \Pi, U)$ in the input list such that $N(Q_y) = U$. If such a tuple exists, we should first check that $x \notin Q_y$. If so, we now have two components of $G \sm U$: $Q_x$ and $Q_y$. By \Cref{lem:volume-checksum}, the only step left to determine whether $U$ is a $k$-shredder is verifying that $\vol(Q_x) + \vol(Q_y) + |E(U, U)| < m$. Since $U$ is a $k$-shredder, this inequality must be true. We can conclude that $U$ will be added to $\L$ by lines \ref{alg:high-deg:U-exists}-\ref{alg:high-deg:add-U-to-L}. Notice that we do not need to explicitly compute $Q_x$ as shown in the pseudocode, as we only need the value $\vol(Q_x)$. \Cref{lem:track-local-volume} implies that this value can be computed and returned by \Cref{alg:local}, so we do not need to handle it here.

    To see that every set in $\L$ is a $k$-shredder, notice that we only append to $\L$ on line \ref{alg:high-deg:add-U-to-L}. When we append a set $U$ to $\L$, we must have found two components $Q_x, Q_y$ of $G \sm U$ such that $\vol(Q_x) + \vol(Q_y) + |E(U, U)| < m$. By \Cref{lem:volume-checksum} we have that $U$ is a $k$-shredder, which implies that every set in $\L$ is a $k$-shredder.
\end{proof}

\begin{lem} \label{lem:high-deg-time}
    \Cref{alg:high-deg} runs in $\O(k^2 m \log^2 n)$ time.
\end{lem}

\begin{proof}
    For each sampled vertex $y$ with volume parameter $\nu$, we perform a restricted BFS to explore a set of vertices $Q_y$. This takes $\O(\nu)$ time. Checking whether a tuple $(x, \nu, \Pi, U)$ exists in $\U$ such that $N(Q_y) = U$ can be done in $\O(k)$ time by hashing the tuples in $\U$ beforehand and querying the hash of $N(Q_y)$. Computing $|E(U, U)|$ can be done in $\O(k^2)$ time using adjacency sets. Finally, computing $\vol(Q_x) + \vol(Q_y) + |E(U, U)|$ can be done in $\O(1)$. In total, these steps take $\O(k^2 + \nu)$ time. Since we vary $\nu$ over all integers from 0 to $\left\lceil \log \left(\frac{m}{k}\right) \right\rceil$, there are $\O(\log m) = \O(\log n)$ values of $\nu$. Putting all the steps together, we can bound the total runtime as
    \begin{align*}
        \sum_{\nu} 400 \frac{m}{\nu} \log n \cdot \O(k^2 + \nu) 
        = & \sum_{i=0}^{\lceil \log (\frac{m}{k}) \rceil} 400 \frac{m}{2^i} \log n \cdot \O(k^2 + 2^i) \\
        = & \sum_{i=0}^{\lceil \log (\frac{m}{k}) \rceil} \O \left(k^2 \frac{m}{2^i} \log n + 400 m \log n \right) \\
        = & \sum_{i=0}^{\lceil \log (\frac{m}{k}) \rceil} \O(k^2 m \log n) \\
        = & \left( \left\lceil \log \left( \frac{m}{k} \right) \right\rceil + 1 \right) \cdot \O(k^2 m \log n) \\
        = & \: \O(\log m) \cdot \O(k^2 m \log n) \\
        = & \: \O(k^2 m \log^2 n).
    \end{align*}
\end{proof}

\begin{proof}[Proof of \Cref{lem:high-deg}]
    Correctness and time complexity of \Cref{alg:high-deg} directly follow from \Cref{lem:high-deg-correctness,lem:high-deg-time}.
\end{proof}

%% file: src/8,9-capture,listing.tex
\section{Capturing and Listing Unbalanced \fmt{$k$}--Shredders} \label{sec:unbalanced}
In the previous sections, we showed an algorithm that takes as input a tuple $(x, \nu, \Pi)$, and lists all $k$-shredders that are captured by $(x, \nu, \Pi)$ as well as an unverified set. We then showed how to resolve unverified sets using casework on the structural properties of $k$-shredders. There is one piece of the puzzle left: the method for capturing unbalanced $k$-shredders. For this, we will leverage random sampling using \Cref{lem:hitting-set} as well as recent developments in local flow algorithms as in \cite{FNSYY20}.

\subsection{Leveraging Local Flow Algorithms}
Let $S$ be an unbalanced $k$-shredder with partition $(C, \R)$. At a high level, we use geometric sampling to obtain a seed vertex $x \in \R$ and a volume parameter $\nu$ satisfying $\frac{1}{2} \nu < \vol(\R) \leq \nu$. This obtains two items necessary for capturing $S$. We are still missing a set of $k$ openly-disjoint paths $\Pi$, each starting from $x$ and ending at a vertex in $C$, such that the sum of lengths over all paths is at most $k^2 \nu$. This is precisely a core tool developed in \cite{FNSYY20}. Below is a theorem stating the existence of a local flow algorithm that precisely constructs our desired set of paths.

\begin{definition}[Vertex Cut]
    A vertex cut $(L, S, R)$ of a graph $G = (V, E)$ is a partition of $V$ such that for all vertex pairs $(u, v) \in L \times R$, $u$ is not connected to $v$ in $G \sm S$.
\end{definition}

\begin{thm}[{\cite[implicit in Theorem 4.1]{FNSYY20}}] \label{thm:LocalVC}
    Let $G = (V, E)$ be an undirected $n$-vertex $m$-edge graph with vertex connectivity $k$. Let $(L, S, R)$ be a vertex cut of $G$ such that $\vol(R) < \frac{m}{k}$. There exists a randomized algorithm $\LVC(\cdot, \cdot)$ that takes as input a pair $(x, \nu)$ where $x$ is a vertex in $R$ and $\nu$ is a positive integer satisfying $\frac{1}{2} \nu < \vol(R) \leq \nu$. The algorithm outputs a set $\Pi$ of $k$ openly-disjoint paths such that each path satisfies the following.
    \begin{enumerate}
        \item The sum of lengths over all paths in $\Pi$ is at most $k^2 \nu$.
        \item The path starts from $x$ and ends at a vertex in $L$.
    \end{enumerate}
    The algorithm outputs $\Pi$ with probability $1/4$ in $\O(k^2 \nu)$ time.
\end{thm}
Theorem 4.1 of \cite{FNSYY20} only states that their algorithm returns a vertex cut. But they also construct the set of paths $\Pi$. Their algorithm is simple, and we will briefly explain it here. First, we perform the standard vertex-splitting reduction and reduce the problem to finding \emph{directed edge-disjoint} paths instead.
To find the first path, we use DFS starting from a vertex $x$ and explore $k \nu$ volume. Then, we sample a random endpoint $y_1$ among all explored edges. Note that $y_1$ is in $L$ with probability at least $1-1/k$ since $\vol(R) \leq \nu$. We reverse the direction of edges on the path from $x$ to $y_1$ in the DFS tree and obtain a ``residual'' graph. Then, we repeat the process in the residual graph to construct the next path from $x$ to $y_2$. After $k$ iterations, the endpoints $y_1, \dots, y_k$ of these $k$ paths are in $L$ with probability $(1-1/k)^k \geq 1/4$. The paths in the residual graphs can be decomposed via flow decomposition into $k$ directed edge-disjoint paths in the original graph whose total length is $k^2 \nu$. Finally, these paths correspond to $k$ openly-disjoint paths by the standard reduction in the beginning.

The output of the algorithm described in \Cref{thm:LocalVC} directly corresponds to \Cref{def:capture}. Let $R$ denote the union of all components in $\R$. Notice that $(C, S, R)$ forms a vertex cut such that $\vol(R) < \frac{m}{k}$. The idea is to obtain a seed vertex $x \in R$ using a linear amount of random samples. For each sample, we can directly apply $\LVC(\cdot, \cdot)$ to obtain the desired set of paths. Furthermore, we can boost the success rate of the algorithm by repeating it a polylogarithmic number of times. In the following section, we formalize this idea.

\subsection{The Algorithm for Unbalanced \fmt{$k$}--Shredders}
The main result is stated below.

\begin{lem} \label{lem:unbalanced}
    There exists a randomized algorithm that takes as input $G = (V, E)$, an $n$-vertex $m$-edge undirected graph with vertex connectivity $k$. The algorithm correctly lists all unbalanced $k$-shredders of $G$ with probability $1 - n^{-98}$ in $\O(k^4 m \log^4 n)$ time.
\end{lem}

We first give a high level outline for listing unbalanced $k$-shredders. Let $S$ be an unbalanced $k$-shredder with partition $(C, \R)$. With geometric sampling and \Cref{lem:hitting-set}, we will sample a vertex $x \in \R$ with a volume parameter $\nu$ that satisfies $\frac{1}{2} \nu < \vol(\R) \leq \nu$. Then, we use \Cref{thm:LocalVC} to obtain a set of $k$ openly-disjoint paths $\Pi$, each starting from $x$ and ending at a vertex in $C$ such that the sum of lengths over all paths is at most $k^2 \nu$. We now have a tuple $(x, \nu, \Pi)$ that captures $S$. After capturing $S$, we call \Cref{alg:local} to list $S$ as a $k$-shredder or an unverified set. In the latter case, we can verify whether $S$ is a low-degree $k$-shredder using \Cref{alg:low-deg}. If this verification step fails, $S$ must be a high-degree $k$-shredder. In this case, we can add $S$ to a global list of unverified sets. This list will be processed after all unbalanced $k$-shredders have been captured. Lastly, we can extract all high-degree $k$-shredders from the list of unverified sets using \Cref{alg:high-deg}. Pseudocode describing this process is given in \Cref{alg:unbalanced}.
 
\begin{algorithm}[!ht]
\caption{Listing all unbalanced $k$-shredders of a graph}
\label{alg:unbalanced}
\begin{algorithmic}[1]
    \Require{$G = (V, E)$ - an undirected $k$-vertex-connected graph}
    \Ensure{$\L$ - a list containing all unbalanced $k$-shredders of $G$}

    \State $\L \leftarrow \varnothing$
    \State $\U \leftarrow \varnothing$
    \State $\Phi \leftarrow$ initialize a pairwise connectivity oracle as in \cite{Kos23} \label{alg:unbalanced:init-oracle}

    \item[]
    \For {$i \leftarrow 0$ to $\lceil \log \left(\frac{m}{k}\right) \rceil$} \label{alg:unbalanced:nu-loop}
        \State $\nu \leftarrow 2^i$
        \For {$400 \frac{m}{\nu} \log n$ times} \label{alg:unbalanced:num-samples}
            \State independently sample an edge $(x, y)$ uniformly at random \label{alg:unbalanced:sample-edge}
            \For {$300 \log n$ times} \label{alg:unbalanced:boost-LVC}
                \State $\Pi \leftarrow \LVC(x, \nu)$ as in \cite{FNSYY20} \label{alg:unbalanced:LVC}
                \State $(\L_{local}, U) \leftarrow$ call \Cref{alg:local} on input $(x, \nu, \Pi)$ \label{alg:unbalanced:call-local}
                \State $\L \leftarrow \L \cup \L_{local}$ \label{alg:unbalanced:add-local}
                \If {$U \neq \varnothing$}
                    \State call \Cref{alg:low-deg} on input $(x, \nu, \Pi, U)$ \label{alg:unbalanced:call-low}
                    \If {\Cref{alg:low-deg} returned $\TRU$}
                        \State $\L \leftarrow \L \cup \{U\}$ \label{alg:unbalanced:add-low}
                    \Else
                        \State $\U \leftarrow \U \cup \{(x, \nu, \Pi, U)\}$ \label{alg:unbalanced:add-U}
                    \EndIf
                \EndIf
            \EndFor
        \EndFor
    \EndFor
    
    \item[]
    \State $\L_{high-degree} \leftarrow$ call \Cref{alg:high-deg} on input $\U$ \label{alg:unbalanced:call-high}
    \State $\L \leftarrow \L \cup \L_{high-degree}$ \label{alg:unbalanced:add-high}
    \State \textbf{return} $\L$
\end{algorithmic}
\end{algorithm}

\begin{lem} \label{lem:unbalanced-capture}
    Let $S$ be an unbalanced $k$-shredder. Then $S$ is captured by a tuple $(x, \nu, \Pi)$ after line \ref{alg:unbalanced:LVC} during some iteration of the loop on line \ref{alg:unbalanced:nu-loop} with probability $1 - n^{-99}$.
\end{lem}

\begin{proof}
    Suppose $S$ has partition $(C, \R)$. Consider the iteration of the loop on line \ref{alg:unbalanced:nu-loop} where $\frac{1}{2} \nu < \vol(\R) \leq \nu$. By \Cref{lem:hitting-set}, we will sample an edge $(x, y)$ such that $x \in \R$ with probability $1 - n^{-100}$ during the \textbf{for} loop on line \ref{alg:unbalanced:num-samples}. The only missing piece of the tuple is the set of paths. This piece is effectively solved by \Cref{thm:LocalVC}. Let $R$ be the union of all components in $\R$. Notice that $(C, S, R)$ is a vertex cut such that $\vol(R) < \frac{m}{k}$. We can call $\LVC(x, \nu)$ to obtain a set $\Pi$ of $k$ openly-disjoint paths that start from $x$ and end in $C$ such that the sum of lengths over all paths is at most $k^2 \nu$ with probability $1/4$. This probability is boosted by the \textbf{for} loop on line \ref{alg:unbalanced:boost-LVC}. So the probability that we fail to obtain such a set of paths over all trials is reduced to $(3/4)^{300 \log n} \leq (1/2)^{100 \log n} = n^{-100}$. Hence, the joint probability that we sample a vertex $x \in \R$ and obtain the desired set of paths $\Pi$ is at least $(1 - n^{-100})^2 \geq 1 - 2n^{-100} \geq 1 - n^{-99}$ for all $n \geq 2$.
\end{proof}

\begin{lem} \label{lem:unbalanced-list-all}
    Let $S$ be a $k$-shredder captured by the tuple $(x, \nu, \Pi)$ after line \ref{alg:unbalanced:LVC}. If $S$ is low-degree, then $S \in \L$ after line \ref{alg:unbalanced:add-low}. Otherwise, if $S$ is high-degree, then $S \in \L$ after line \ref{alg:unbalanced:add-high} with probability $1 - n^{-100}$.
\end{lem}

\begin{proof}
    Suppose $S$ has partition $(C, \R)$ and is captured by $(x, \nu, \Pi)$. If \Cref{alg:local} returns $S$ as a $k$-shredder, we are immediately done by lines \ref{alg:unbalanced:call-local}-\ref{alg:unbalanced:add-local}. So let us assume that $S$ is reported as an unverified set. We have that $\frac{1}{2}\nu < \vol(\R) \leq \nu$. If $S$ is low-degree, we have that there exists a vertex $s \in S$ such that $\deg(s) \leq \nu$. By \Cref{lem:low-deg-aux}, \Cref{alg:low-deg} will return $\TRU$ on input $(x, \nu, \Pi, S)$. Then, $S$ is added to $\L$ by line \ref{alg:unbalanced:add-low}. Otherwise, suppose that $S$ is high-degree. Then, $(x, \nu, \Pi, S)$ is added to the set $\U$ by line \ref{alg:unbalanced:add-U}. By \Cref{lem:high-deg}, \Cref{alg:high-deg} will return a list of high-degree $k$-shredders containing $S$ on line \ref{alg:unbalanced:call-high} with probability $1 - n^{-100}$. Hence, $S \in \L$ by line \ref{alg:unbalanced:add-high} with probability $1 - n^{-100}$.
\end{proof}

\begin{lem} \label{lem:unbalanced-list-only-shredders}
    Every set $S \in \L$ at the end of the algorithm is a $k$-shredder.
\end{lem}

\begin{proof}
    There are only three lines of the algorithm where we add sets to $\L$: \ref{alg:unbalanced:call-local}, \ref{alg:unbalanced:add-low}, \ref{alg:unbalanced:add-high}. Line \ref{alg:unbalanced:call-local} only adds $k$-shredders due to \Cref{lem:local}. Line \ref{alg:unbalanced:add-low} only adds $k$-shredders due to \Cref{lem:low-deg-aux}. Lastly, line \ref{alg:unbalanced:add-high} only adds $k$-shredders due to \Cref{lem:high-deg}.
\end{proof}

\begin{lem} \label{lem:unbalanced-time}
    \Cref{alg:unbalanced} runs in $\O(k^4 m \log^4 n)$ time.
\end{lem}

\begin{proof}
    Firstly, line \ref{alg:unbalanced:init-oracle} runs in $\O(km \log n)$ time by \Cref{thm:pairwise-conn-oracle}. We show that lines \ref{alg:unbalanced:LVC}-\ref{alg:unbalanced:add-U} run in $\O(k^4 \log n + k^2 \nu \log \nu)$ time. Line \ref{alg:unbalanced:LVC} runs in $\O(k^2 \nu)$ time according to \Cref{thm:LocalVC}. Line \ref{alg:unbalanced:call-local} (calling \Cref{alg:local}) runs in $\O(k^2 \nu \log \nu)$ time according to \Cref{lem:local}. Line \ref{alg:unbalanced:call-low} (calling \Cref{alg:low-deg}) runs in $\O(k^4 \log n + k\nu)$ time according to \Cref{lem:low-deg-aux}. The other lines can be implemented in $\O(k)$ time using hashing. Now, we bound the total run time of lines \ref{alg:unbalanced:nu-loop}-\ref{alg:unbalanced:add-U} as
    \begin{align*}
        & \sum_{\nu} 400 \frac{m}{\nu} \log n \cdot 300 \log n \cdot \O(k^4 \log n + k^2 \nu \log \nu) \\
        = & \sum_{i=0}^{\lceil \log (\frac{m}{k}) \rceil} \O \left( k^4 \frac{m}{\nu} \log^3 n + k^2 m \log^2 n \log \nu \right) \\
        = & \sum_{i=0}^{\lceil \log (\frac{m}{k}) \rceil} \O(k^4 m \log^3 n) \\
        = & \left( \left\lceil \log \left( \frac{m}{k} \right) \right\rceil + 1 \right) \cdot \O(k^4 m \log^3 n) \\
        = & \: \O(\log m) \cdot \O(k^4 m \log^3 n) \\
        = & \: \O(k^4 m \log^4 n).
    \end{align*}
    Line \ref{alg:unbalanced:call-high} (calling \Cref{alg:high-deg}) runs in $\O(k^2 m \log^2 n)$ time according to \Cref{lem:high-deg} and does not change the overall time complexity.
\end{proof}

\begin{proof}[Proof of \Cref{lem:unbalanced}]
    Let $S$ be an unbalanced $k$-shredder. We have that $S$ is captured by \Cref{alg:unbalanced} with probability $1 - n^{-99}$ by \Cref{lem:unbalanced-capture}. By \Cref{lem:unbalanced-list-all}, if $S$ is captured and low-degree, it is put in $\L$ with probability $1$. Otherwise, if $S$ is high-degree, it is put in $\L$ with probability $1 - n^{-100}$. Therefore, $S \in \L$ with probability at least $(1 - n^{-99}) \cdot (1 - n^{-100}) \geq 1 - n^{-98}$ for all $n \geq 2$. Furthermore, every set in $\L$ is a $k$-shredder by \Cref{lem:unbalanced-list-only-shredders}. Finally, \Cref{alg:unbalanced} runs in $\O(k^4 m \log^4 n)$ time by \Cref{lem:unbalanced-time}.
\end{proof}

\section{Listing All \fmt{$k$}--Shredders} \label{sec:all-k-shredders}
At last, we are ready to present the algorithm for listing all $k$-shredders.

\begin{lem} \label{lem:aks}
    Let $G = (V, E)$ be an $n$-vertex $m$-edge undirected graph with vertex connectivity $k$. There exists a randomized algorithm that takes $G$ as input and correctly lists all $k$-shredders of $G$ with probability $1 - n^{-97}$ in $\O(k^4m\log^4n)$ time.
\end{lem}

Any $k$-shredder can be classified as either balanced or unbalanced. We have not explicitly presented an algorithm for listing balanced $k$-shredders, but it can be found in the appendix (see \Cref{alg:balanced}). Combining the two algorithms, we can lists balanced and unbalanced $k$-shredders by calling them sequentially.

\begin{algorithm}[!ht]
\caption{Listing all $k$-shredders of a graph} \label{alg:all}
\begin{algorithmic}[1]
    \Require{$G = (V, E)$ - an undirected $k$-vertex-connected graph}
    \Ensure{$\L$ - a list of all $k$-shredders of $G$}

    \State $\L \leftarrow \varnothing$
    \item[]
    
    \State $\L_{balanced} \leftarrow$ call \Cref{alg:balanced} on input $G$
    \State $\L \leftarrow \L \cup \L_{balanced}$ \label{alg:all:add-balanced}

    \item[]
    \State $\L_{unbalanced} \leftarrow$ call \Cref{alg:unbalanced} on input $G$
    \State $\L \leftarrow \L \cup \L_{unbalanced}$ \label{alg:all:add-unbalanced}

    \item[]
    \State \textbf{return} $\L$
\end{algorithmic}
\end{algorithm}

\begin{lem} \label{lem:aks-correctness}
    Let $S$ be a $k$-shredder. Then, $S$ is returned by \Cref{alg:all} with probability $1 - n^{-98}$. Furthermore, \Cref{alg:all} only returns $k$-shredders.
\end{lem}

\begin{proof}
    We use proof by casework. If $S$ is balanced, then \Cref{lem:balanced} implies $S$ is listed by \Cref{alg:balanced} with probability $1 - n^{-100}$ (line \ref{alg:all:add-balanced}). Otherwise, \Cref{lem:unbalanced} implies $S$ is listed by \Cref{alg:unbalanced} with probability $1 - n^{-98}$ (line \ref{alg:all:add-unbalanced}). Hence, with probability $1 - n^{-98}$, $S$ is returned by \Cref{alg:all}.
\end{proof}

\begin{lem} \label{lem:aks-union-bound}
    All $k$-shredders are listed by \Cref{alg:all} with probability $1 - n^{-97}$.
\end{lem}

\begin{proof}
    Jord\'an proved that there are at most $n$ $k$-shredders in \cite{Jor95}. \Cref{lem:aks-correctness} implies that each $k$-shredder is listed with probability $1 - n^{-98}$. This means that the probability of any one $k$-shredder fails to be listed is at most $n^{-98}$. Using the union bound, we have that the probability that any $k$-shredder fails to be listed is at most $n \cdot n^{-98} = n^{-97}$. Therefore, with probability $1 - n^{-97}$, all $k$-shredders are listed.
\end{proof}

\begin{lem} \label{lem:aks-time}
    \Cref{alg:all} runs in $\O(k^4 m \log^4 n)$ time.
\end{lem}

\begin{proof}
    Line \ref{alg:all:add-balanced} runs in $\O(k^2 m \log n)$ time by \Cref{lem:balanced}. Line \ref{alg:all:add-unbalanced} runs in $\O(k^4 m \log^4 n)$ by \Cref{lem:unbalanced}. The running time of the last step dominates.
\end{proof}

\begin{proof}[Proof of \Cref{lem:aks}]
    The correctness of \Cref{alg:all} is given by \Cref{lem:aks-correctness,lem:aks-union-bound}. The time complexity is given by \Cref{lem:aks-time}.
\end{proof}

We have now shown and proved the existence of a randomized algorithm that improves upon the runtime of $\AKS(\cdot)$. The final step is to include the preprocessing step of sparsifying the graph.

\thmaks*

\begin{proof}
    As stated earlier, the algorithm developed in \cite{NI92} can be used to preprocess $G$ such that all $k$-shredders are preserved and $m$ is reduced to at most $\O(kn)$ edges. This follows from \cite[Proposition 3.3]{CT99}. We can call the sparsifying algorithm and then \Cref{alg:all}. The correctness of this procedure directly follows from \cite[Proposition 3.3]{CT99} and \Cref{lem:aks}.
    For time complexity, the sparsifying algorithm takes $\O(m)$ time. Afterwards, we can use \Cref{alg:all} to list all $k$-shredders. This takes $\O(k^4m\log^4n) = \O(k^4(kn)\log^4n) = \O(k^5n\log^4n)$ time. In total, both steps take $\O(m + k^5n\log^4n)$ time. 
\end{proof}

%% file: src/10-shattering.tex
\section{Most Shattering Minimum Vertex Cut} \label{sec:most-shattering}
We have presented a randomized near-linear time algorithm that lists all $k$-shredders with high probability, but we have not yet shown how to count the number of components each $k$-shredder separates. Counting components proves to be a bit trickier than it was in \cite{CT99}. One may hope for the existence of an oracle that takes a set $S$ of $k$ vertices as input and counts the number of components in $G \sm S$ in $\tO(\poly(k))$ time. If there exists such an oracle, we are immediately finished because there are at most $n$ $k$-shredders. We can simply query such an oracle for each $k$-shredder to spend $\tO(\poly(k) \cdot n)$ time in total. Unfortunately, such a data structure cannot exist assuming SETH, as shown in \cite[Theorem 8.9]{LS22}.

What we will do is modify our algorithms such that whenever we list a $k$-shredder $S$, we also return the number of components of $G \sm S$. We will do this for \Cref{alg:local}, \Cref{alg:low-deg}, and \Cref{alg:high-deg}. We do not need to modify \Cref{alg:balanced} because it lists $k$-shredders only by calling $\SHR(\cdot, \cdot)$ as a subroutine, and $\SHR(\cdot, \cdot)$ already supports counting the number of components as proved in \cite{CT99}. Note that if there are no $k$-shredders of $G$, we can simply return a $k$-separator of $G$. This can be done in $\tO(m + k^3n)$ time as shown in \cite{FNSYY20}.

\subsection{Counting Components in the Local Algorithm}
\begin{lem} \label{lem:local-mod}
    Suppose that \Cref{alg:local} returns $(\L, U)$ on input $(x, \nu, \Pi)$. There exists a modified version of \Cref{alg:local} such that for each $k$-shredder $S \in \L$ that is captured by $(x, \nu, \Pi)$, the modified version also computes the number of components of $G \sm S$. The modification requires $\O(k^2 \nu)$ additional time, subsumed by the running time of \Cref{alg:local}.
\end{lem}

\begin{proof}
    Suppose \Cref{alg:local} identifies a candidate $k$-shredder $S$ on line \ref{alg:local:mark-candidate}. By line \ref{alg:local:not-term-early}, all bridges of $\Pi$ attached to some vertex $s \in S$ must have been explored. The modification is as follows. We can keep a dictionary $M$ whose keys are candidate $k$-shredders (in $k$-tuple form) mapped to nonnegative integers initialized to zero. For each bridge $\Gamma$ of $\Pi$ attached to $s$, if $\Gamma$ is a component of $G \sm \Pi$ such that $N(\Gamma) = S$, then we increment $M[S]$. We prove that if $S \in \L$ is a $k$-shredder captured by $(x, \nu, \Pi)$, then $M[S] + 2$ is the number of components of $G \sm S$.

    Let $S \in \L$ be a $k$-shredder with partition $(C, \R)$ captured by $(x, \nu, \Pi)$. Let $Q_x$ denote the component in $\R$ containing $x$. From \Cref{lem:preserve-components}, we have that every component in $\R \sm \{Q_x\}$ is a component of $G \sm \Pi$. This means every component in $\R \sm \{Q_x\}$ is a bridge of $\Pi$. Furthermore, if we fix a vertex $s \in S$, notice that all components in $\R \sm \{Q_x\}$ are bridges of $\Pi$ attached to $s$. These components are special; each component of $R \sm \{Q_x\}$ is a bridge $\Gamma$ of $\Pi$ that satisfies $N(\Gamma) = S$. What our modification does is to count the number of bridges $\Gamma$ attached to $s$ that satisfy $N(\Gamma) = S$. This means that after processing all bridges attached to $s$, we have $M[S] = |\R \sm \{Q_x\}|$. There are only two components of $G \sm S$ that were not counted: $Q_x$ and $C$. Hence, there are $M[S] + 2$ components of $G \sm S$.

    The time complexity for this operation is straightforward. For each candidate $k$-shredder $S$, we spend $\O(k)$ time to hash and store it in $M$. For each bridge $\Gamma$ of $\Pi$, we need to check whether $N(\Gamma)$ is a key in the dictionary. This can be done by hashing $N(\Gamma)$ in $\O(k)$ time. As proved in \Cref{lem:local:time}, there are at most $\O(\nu)$ candidate $k$-shredders, and at most $\O(k \nu)$ explored bridges. Hence, hashing all candidate $k$-shredders and bridges can be done in $\O(k) \cdot \O(\nu + k \nu) = \O(k^2 \nu)$ time. We conclude that these modifications impose an additional time cost of $\O(k^2 \nu)$, which is subsumed by the runtime of the original algorithm.
\end{proof}

\subsection{Counting Components for Low-Degree Unverified Sets}
If \Cref{alg:low-deg} returned $\TRU$ on input $(x, \nu, \Pi, U)$, then $U$ is a $k$-shredder with a vertex $u \in U$ such that $\deg(u) \leq \nu$. Notice that all components of $G \sm U$ must contain a vertex that is adjacent to $u$. Immediately, we have that the number of components of $G \sm U$ is upper bounded by $\nu$, as $\deg(u) \leq \nu$. However, looking at two arbitrary edges $(u, v_1), (u, v_2)$ adjacent to $u$, it is unclear whether $v_1$ and $v_2$ are in the same component of $G \sm U$. Although we can query whether $v_1$ and $v_2$ are connected in $G \sm U$ in $\O(k)$ time, we cannot afford to make these queries for all pairs of vertices in $N(u)$. Such a procedure would require us to make $\O(\nu^2)$ many queries, which is too costly. 

To solve this issue, we must explain some context regarding the pairwise connectivity oracle developed in \cite{Kos23}. Essentially, the oracle is a tree $T$ obtained by running a DFS traversal of $G$. It makes use of the following useful property of DFS-trees.

\begin{fact}[\cite{Tar72}] \label{fact:back-edges}
    Let $T$ be a DFS-tree of $G$, rooted arbitrarily. Then, for every edge $(u, v)$ of $G$, $u$ and $v$ are related as ancestor and descendant in $T$.
\end{fact}

Edges that are omitted from $T$ are called \textit{back edges}, as its endpoints must be a vertex and one of its ancestors discovered earlier in the traversal. Let $U$ be a set of $k$ vertices. Given a component $C$ of $T \sm U$, we define $r_C$ as the root vertex of $C$. A component $C$ of $T \sm U$ is called an \textit{internal component} if $r_C$ is the ancestor of a vertex in $U$. Otherwise, $C$ is called a \textit{hanging subtree}. An important observation is that because $|U| = k$, there can be at most $k$ internal components. The oracle exploits this fact by reducing pairwise connectivity queries in $G \sm U$ to connectivity queries between the internal components of $T \sm U$. To be concrete, we prove two auxiliary lemmas to prepare us for counting the number of components of $G \sm U$.

\begin{lem} \label{lem:T-conn}
    Let $T$ be a DFS-tree of $G$ and $U$ be a vertex set. For a vertex $v \notin U$, let $C_v$ denote the component of $T \sm U$ containing $v$. Two vertices $x$ and $y$ in $V \sm U$ are connected in $G\setminus U$ if and only if one of the following is true.
    \begin{enumerate}
        \item $C_x = C_y$.
        \item There exists an internal component $C$ of $T \sm U$ that both $x$ and $y$ are connected to in $G \sm U$.
    \end{enumerate}
\end{lem}

\begin{proof}
    We prove the forward direction (the backward direction is trivial). Suppose we are given vertices $x$ and $y$ and wish to return whether $x$ and $y$ are connected in $G \sm U$. If $C_x = C_y$, then $x$ and $y$ are trivially connected in $G \sm U$. Otherwise, there must be a path from $x$ to $y$ in $G \sm U$ using back edges. Since the endpoints of back edges are always related as ancestor and descendant, this rules out the existence of edges between distinct hanging subtrees of $T \sm U$. This is because two vertices in two distinct hanging subtrees cannot be related as ancestor and descendant. Hence, $x$ and $y$ must both be connected to an internal component of $T \sm U$.
\end{proof}

\begin{lem} \label{lem:internal-comps}
    Let $T$ be a DFS tree and let $U$ be a set of $k$ vertices. There exists an algorithm that takes $U$ as input and returns a list of vertices $\I$ that satisfies the following. For each internal component $C$ of $T \sm U$, there is exactly one vertex $v \in \I$ such that $v \in C$. Additionally, every vertex in $\I$ is in an internal component of $T \sm U$. The algorithm runs in time $\O(k^2)$.
\end{lem}

This lemma is proved and follows from \cite[Section 3.3]{Kos23}. The idea is to query for each pair $(u, v) \in U \times U$ whether $u$ is an ancestor or descendant of $v$. This can be done using a global pre-ordering and post-ordering of the vertices as a preprocessing step in $\O(n + m)$ time to support $\O(1)$ queries. If $u$ is an ancestor of $v$ and no other vertex in $U$ is a descendant of $u$ and an ancestor of $v$, then we say $u$ ($v$) is an immediate ancestor (descendant) of $v$ ($u$). The key observation is that the internal components of $T \sm U$ must live between direct ancestor/descendant pairs of $U$. Now we are ready to prove the following lemma.

\begin{lem} \label{lem:low-deg-mod}
    Suppose that \Cref{alg:low-deg} returns $\TRU$ on input $(x, \nu, \Pi, U)$. There exists a modified version of \Cref{alg:low-deg} that also computes the number of components of $G \sm U$. The modification requires $O(k^2\nu)$ additional time, subsumed by the running time of \Cref{alg:low-deg}.
\end{lem}

\begin{proof}
    We can remodel the problem of counting the components of $G \sm U$ as the following. Fix a vertex $u \in U$, and suppose that $N(u) = \{y_1, \ldots, y_{|N(u)|}\}$. For each vertex $y_i$, we wish to determine whether any vertex in $\{y_1, \ldots, y_{i-1}\}$ is connected to $y_i$. Let $q$ denote the number of vertices $y_i$ such that $y_i$ is not connected to any vertex in $\{y_1, \ldots, y_{i-1}\}$. It is straightforward to see that $q$ is the number of components of $G \sm U$. The idea is to use a DFS-tree $T$ and calls to the pairwise connectivity oracle to count $q$. Consider \Cref{alg:count-low-deg-comps}.

    \begin{algorithm}[!ht]
    \caption{Counting components for low-degree unverified sets}
    \label{alg:count-low-deg-comps}
    \begin{algorithmic}[1]
        \Require{$(x, \nu, \Pi, U)$ - a tuple that \Cref{alg:low-deg} returned $\TRU$ on}
        \Ensure{the number of components of $G \sm U$}

        \State $T \leftarrow$ a global DFS-tree initialized prior to calling this algorithm
        \State $\I \leftarrow$ representatives for internal components of $T \sm U$ (from \Cref{lem:internal-comps}) \label{alg:count:get-internals}
        \State $u \leftarrow$ a vertex in $U$ with $\deg(u) \leq \nu$ \label{alg:count:get-low-deg}

        \item[]
        \State $\tt{counter} \leftarrow 0$
        \For {each edge $(u, y)$ adjacent to $u$ such that $y \notin U$} \label{alg:count:edge}
            \State $r_y \leftarrow$ the root of the component of $T \sm U$ containing $y$ \label{alg:count:root}
            \If {$r_y$ is marked as seen} \label{alg:count:seen-root}
                \State \textbf{skip} to next edge
            \EndIf
            \State mark $r_y$ as seen \label{alg:count:mark-root}

            \item[]
            \State $\tt{flag} \leftarrow \TRU$
            \For {each vertex $v \in \I$} \label{alg:count:internals}
                \If {$y$ and $v$ are not connected in $G \sm U$} \label{alg:count:y-v-conn}
                    \State \textbf{skip} to next vertex
                \EndIf
                \If {$v$ is marked as seen} \label{alg:count:seen-internal}
                    \State $\tt{flag} \leftarrow \FLS$
                    \State \textbf{break}
                \Else
                    \State mark $v$ as seen \label{alg:count:mark-internal}
                \EndIf
            \EndFor

            \item[]
            \If {$\tt{flag}$ is $\TRU$}
                \State $\tt{counter} \leftarrow \tt{counter} + 1$ \label{alg:count:increment}
            \EndIf
        \EndFor
    
        \State \textbf{return} $\tt{counter}$
    \end{algorithmic}
    \end{algorithm}

    Now we prove the correctness of \Cref{alg:count-low-deg-comps}. Suppose we are processing an edge $(u, y)$ during the \textbf{for} loop on line \ref{alg:count:edge}. We wish to determine whether there exists an edge $(u, y')$ processed in a previous iteration such that $y$ and $y'$ are connected in $G \sm U$. Assume that such an edge exists. Let $(u, y')$ be the earliest edge processed in a previous iteration such that $y$ and $y'$ are connected in $G \sm U$. Let $C_y, C_{y'}$ denote the components containing $y, y'$ in $T \sm U$, respectively. \Cref{lem:T-conn} implies that either $C_y = C_{y'}$, or $y$ and $y'$ are both connected to an internal component of $T \sm U$. If $C_y = C_{y'}$, then we would have marked the root of $C_y$ as seen when processing $(u, y')$ on line \ref{alg:count:mark-root}. Otherwise, $y$ and $y'$ are both connected to an internal component $C$ of $T \sm U$. Let $v$ be the unique vertex in $\I \cap C$, which we know must exist by \Cref{lem:internal-comps} and line \ref{alg:count:get-internals}. During the \textbf{for} loop on line \ref{alg:count:internals}, we will see that $v$ was marked as seen on line \ref{alg:count:seen-internal}. This is because when processing $(u, y')$, we must have eventually marked $v$ as seen on line \ref{alg:count:mark-internal}. Conversely, suppose that there does not exist a prior edge $(u, y')$ such that $y$ and $y'$ are connected in $G \sm U$. Again, \Cref{lem:T-conn} implies that $C_y$ should not be marked as seen and no internal component connected to $y$ in $G \sm U$ should be marked as seen. Otherwise, it would contradict our assumption. Therefore, $\tt{flag}$ is $\TRU$ by line \ref{alg:count:increment} and we correctly increment the component counter.

    Now we prove the time complexity of the algorithm. Line \ref{alg:count:get-internals} runs in $\O(k^2)$ time by \Cref{lem:internal-comps}. The main \textbf{for} loop on line \ref{alg:count:edge} runs for $\O(\nu)$ iterations. Line \ref{alg:count:root} can be implemented in $\O(k)$ time as seen in \cite{Kos23}. The \textbf{for} loop on line \ref{alg:count:internals} runs for $\O(k)$ iterations. For each iteration, we need to query whether $y$ and $v$ are connected in $G \sm U$ or whether $v$ has been marked as seen. This can be done in $\O(k)$ by querying the pairwise connectivity oracle in \Cref{thm:pairwise-conn-oracle}. Hence, lines \ref{alg:count:y-v-conn}-\ref{alg:count:mark-internal} can be implemented in $\O(k)$. This means lines \ref{alg:count:internals}-\ref{alg:count:mark-internal} runs in $\O(k^2)$ time. Tallying all the steps up, the entire modification runs in $\O(k^2 + k + \nu(k + k^2)) = \O(k^2 \nu)$ time additional to the time takes by \Cref{lem:low-deg-aux:time}.
\end{proof}

\subsection{Counting Components for High-Degree Unverified Sets}
\begin{lem} \label{lem:high-deg-mod}
    Suppose that \Cref{alg:high-deg} returns a list $\L$ of $k$-shredders on input $\U$. There exists a modified version of \Cref{alg:high-deg} such that for each $k$-shredder $S \in \L$, the modified version also computes the number of components of $G \sm S$. The modification requires $\O(m \log^2 n)$ additional time, subsumed by the time complexity of \Cref{alg:high-deg}.
\end{lem}

\begin{proof}
    Recall that \Cref{alg:high-deg} uses sampling to list high-degree $k$-shredders. Fix a high-degree $k$-shredder $S$ with partition $(C, \R)$. Recall from \Cref{sec:high-deg} that the key observation is that $S$ forms a wall of high-degree vertices. We exploited this by sampling vertices in $\R$ and exploring regions of vertices with low degree. The idea is that $S$ will restrict the graph exploration within one component of $\R$, and we can recover that component. We can actually extend this idea a bit further to count the number of components in $\R$. Let $Q$ be an arbitrary component of $\R$. Let $\nu$ be the power of two satisfying $\frac{1}{2}\nu < \vol(Q) \leq \nu$. By \Cref{lem:hitting-set}, we can sample $400 \frac{m}{\nu} \log n$ edges to obtain an edge $(y, z)$ such that $y \in Q$ with high probability. Since $\vol(\R) < \frac{m}{k}$, we have $\vol(Q) < \frac{m}{k}$. So, if we vary $\nu$ over all powers of two up to $\frac{m}{k}$ and repeat the sampling scheme, we will eventually obtain a vertex $y \in Q$ for all components $Q \in \R$ with high probability. Notice that we have not modified \Cref{alg:high-deg} yet, the probability analysis already holds for the algorithm as it currently stands.

    The modification is as follows. We can keep a dictionary $M$ whose keys are the unverified sets in $\U$ (in $k$-tuple form) mapped to sets of vertices initialized to empty sets. We can think of $M[U]$ as a set of vertices that are each in a distinct component of $G \sm U$. The idea is to add vertex samples to $M[U]$ one by one, collecting representatives in distinct components of $G \sm U$. We add a few steps after line \ref{alg:high-deg:U-exists}. First, we check whether any vertices in $Q$ are in the set $M[U]$. If so, this means that the component of $G \sm U$ containing $y$ was already explored during a previous sample. Otherwise, we append $y$ to $M[U]$. We claim that if $U$ is a high-degree $k$-shredder in $\L$ at the end of \Cref{alg:high-deg}, then the number of components of $G \sm U$ is $|M[U]| + 1$. The $+1$ term is due to the fact that the large component $C$ of $G \sm U$ will not be counted because $\vol(C) \gg \frac{m}{k}$.

    To see why this modification works, notice that every time we enter the \textbf{if} branch on line \ref{alg:high-deg:U-exists}, we have found a component $Q$ of $G \sm U$ because $N(Q) = U$. The set $M[U]$ stores a list of vertices each in distinct components of $G \sm U$ that have been explored. Whenever we explore a component $Q$, we simply need to check whether any vertices in $Q$ are in $M[U]$. If not, then it means $Q$ has not been explored before. Therefore, we add $y$ to $M[U]$ as a representative of the component $Q$. We have already showed that at the end of the entire algorithm, we will have sampled a vertex $x \in Q$ for all components $Q \in \R$ with high probability. This means all components in $\R$ will contain one unique representative in $M[U]$. Hence, $|M[U]| = |\R|$, so the number of components of $G \sm U$ is $|M[U]| + 1$.
    
    As for time complexity, we can use hashing to retrieve the set $M[U]$ given $U = N(Q)$ in $\O(k)$ time, as each unverified set must be exactly $k$ vertices. We can check whether a vertex is in $M[U]$ in $\O(1)$ time. Adding new vertices to a hash set can be done in $\O(1)$. Checking whether any vertex in $Q$ is in $M[U]$ can be done in $\O(\vol(Q))$ time by iterating over all vertices in $Q$ and performing a lookup in $M[U]$. Furthermore, the amount of insertion operations into $M[U]$ over all unverified sets $U$ is upper bounded by the number of vertex samples we use, which is at most $\O(m \log n)$ (recall \Cref{alg:high-deg}). Since we spend $\O(\vol(Q)) = \O(\nu)$ additional time for each sample, the additional cost of this modification is $\O(m \log^2 n)$. We conclude that the cost is subsumed by the running time of the original algorithm.
\end{proof}

\corshatter*
\begin{proof}
    We can keep track of the maximum number of components separated by all $k$-shredders listed so far. Using \Cref{lem:local-mod,lem:low-deg-mod,lem:high-deg-mod}, we can combine all three modifications so that for every $k$-shredder we list, we also keep track of the one maximizing the number of separated components. The proof of time complexity is identical to that of \Cref{lem:aks-time}. The only difference is that \Cref{alg:low-deg} now runs in $\O(k^4 \log n + k^2 \nu)$. However, this does not change the overall time complexity of the algorithm.
\end{proof}

\section{Conclusion and Open Problems} \label{sec:conclusion}
We have presented a Monte Carlo algorithm for finding all the $k$-shredders of a $k$-vertex-connected graph in almost-linear time for fixed values of $k$ and for all $k = \O(\polylog(n))$. This greatly improves upon the strongly-quadratic running time of the previously best known algorithm in \cite{CT99}. In our work, we have shown how to take advantage of recent developments in local flow algorithms \cite{FNSYY20} and connectivity oracles subject to vertex failures \cite{Kos23}. One can hope to extend our algorithms to improve the time bound of the dynamic algorithm that maintains $k$-shredders under edge insertions and deletions in \cite{CT99}. More importantly, it would be interesting if these results could be extended to improve the vertex-connectivity augmentation algorithm suggested by \cite{CT99} and \cite{Jor95}. 

One may also investigate the Pfaffian orientation problem for bipartite graphs. As mentioned earlier, Hegde observed that one of the bottlenecks in the Pfaffian orientation algorithm shown by \cite{RST99} is determining whether a special type of bipartite graph known as a \emph{brace} contains a 4-shredder. However, while a brace on $n \geq 5$ vertices is guaranteed to be 3-connected by \cite[3.2]{RST99}, it is not necessarily 4-connected. It is unlikely that our algorithm can be modified to accomplish this task, as our algorithms for listing $k$-shredders heavily rely on the fact that between any two vertices, there exist $k$ openly-disjoint paths. This is no longer the case in $(k-1)$-connected graphs. Hence, whether there is a near-linear time algorithm for finding a 4-shredder in a brace remains an interesting open problem.

More generally, when $k$ is not a fixed constant, we are unaware of any polynomial time algorithm that finds a $k$-shredder in $(k-1)$-connected graphs. Note that it is possible for a $(k-1)$-connected graph to contain an exponential number of $k$-shredders. Specifically, one can construct $(k-1)$-connected graphs that admit $\Omega(2^{k-1})$ $(k-1)$-separators such that each $(k-1)$-separator can be extended to a $k$-shredder.
\begin{figure}[!ht]
    \centering
    \includegraphics[scale=0.18]{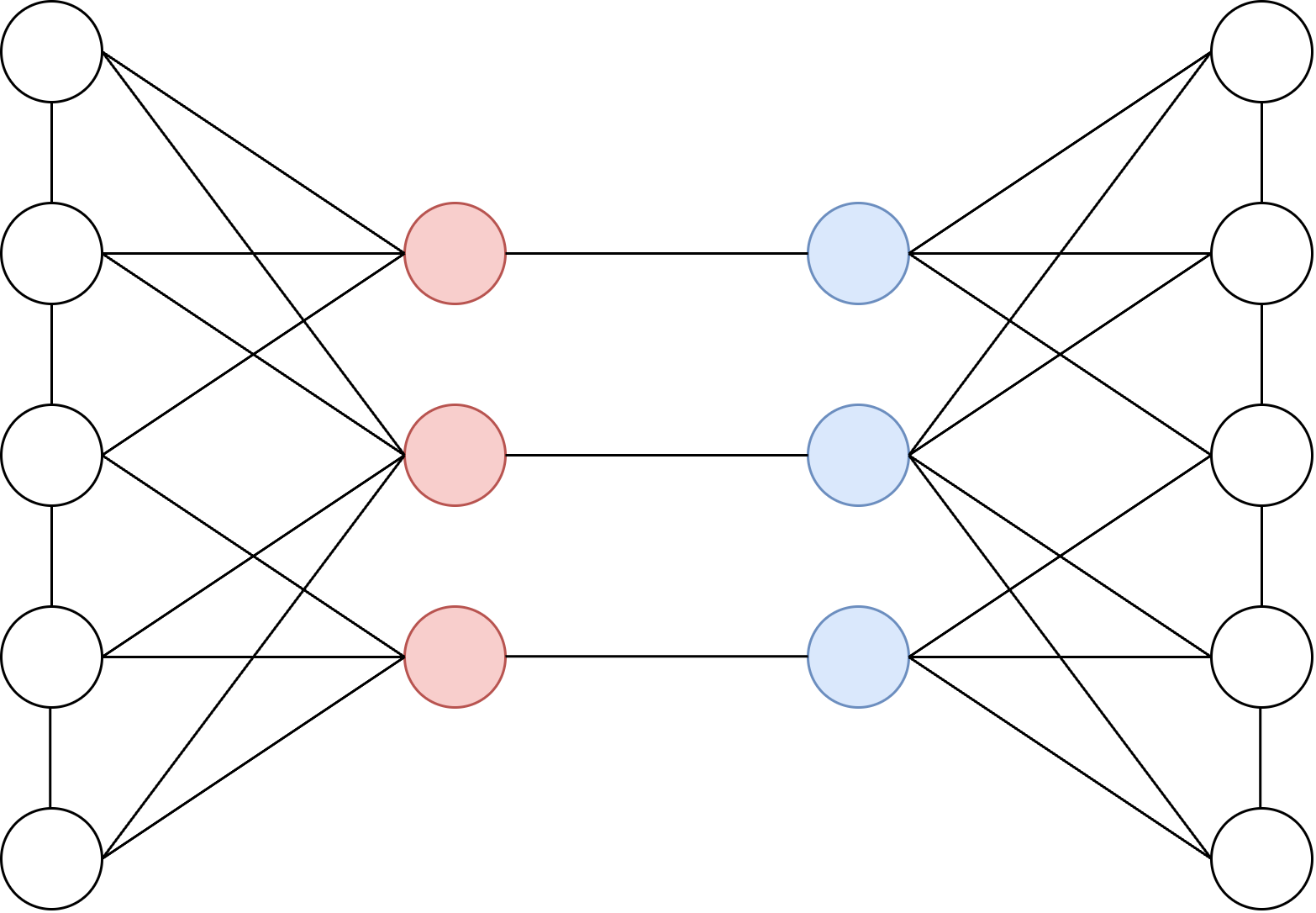}
    \caption{A 3-connected graph that admits $2^3$ 4-shredders.}
    \label{fig:dumbbell}
\end{figure}
For example, the 3-connected graph shown in \Cref{fig:dumbbell} admits at least $2^3$ many 3-separators by picking one red or blue vertex for each of the three middle edges. It can be shown by casework that each 3-separator of this form can be extended to a 4-shredder. While this construction seems to rule out the possibility of listing all $k$-shredders in general $(k-1)$-connected graphs in polynomial time, it is possible that the bipartiteness of braces can be exploited for a lower upper bound and faster algorithm. We leave the problem of removing this bottleneck to the reader.

%% file: src/A-appendix.tex
\section{Proof of Lemma \ref{lem:balanced}} \label{sec:omitted-proofs:2}
\lembalanced*
Fix a balanced $k$-shredder $S$ of $G$ with partition $(C, \R)$. Suppose we sample two edges $(x, x')$ and $(y, y')$ such that $x \in \R$ and $y$ is in a different component from $x$ in $G \sm S$. Then, $\SHR(x, y)$ will list $S$. Consider \Cref{alg:balanced}.

\begin{algorithm}[!ht]
\caption{Listing balanced $k$-shredders}
\label{alg:balanced}
\begin{algorithmic}[1]
    \Require{$G = (V, E)$ - an $n$-vertex $m$-edge undirected graph with vertex connectivity $k$}
    \Ensure{$\L$ - a list containing all balanced $k$-shredders of $G$}
    \For {$800k \log n$ times} \label{alg:balanced:num-samples}
        \State independently sample two edges $(x, x'), (y, y')$ uniformly at random \label{alg:balanced:sample-edge}
        \State $\L \leftarrow \L \cup \SHR(x, y)$ \label{alg:balanced:call-SHR}
    \EndFor
    \State \textbf{return} $\L$
\end{algorithmic}
\end{algorithm}
\begin{proof}[Proof of \Cref{lem:balanced}]    
    Firstly, we can assume that $m \geq 2k^2$. Otherwise, the algorithms given by \cite{CT99} already run in $\O(k^4 n)$ time, which is near-linear for $k = \O(\polylog (n))$. Therefore, if we sample an edge $(x, x')$ uniformly at random, the probability that at least one of $x$ or $x'$ is not in $S$ is at least $1/2$. This follows from the fact that $|E(S, S)| \leq k^2$. Since we randomly pick an endpoint of an edge sampled uniformly at random, $\Pr[x \notin S]$ is at least $1/2 \cdot 1/2 = 1/4$.

    Suppose that we have sampled an edge $(x, x')$ such that $x \notin S$. Let $Q_x$ denote the component of $G \sm S$ containing $x$. For $y$ to be in a different component of $G \sm S$, we need to analyze the probability that $y \in V \sm (Q_x \cup S)$. Notice that because $S$ is balanced, we have that $\vol(V \sm (Q_x \cup S)) \geq \vol(\R)$. This follows from the fact that if $Q_x = C$, then we obtain equality. Otherwise, because $C$ is the largest component by volume, the LHS can only be greater. Hence, we have
    \begin{align*}
        \Pr[y \in V \sm (Q_x \cup S)] & \geq \frac{1}{2} \cdot \frac{\vol(\R)}{m} \\
        & \geq \frac{1}{2k}.
    \end{align*}
    Therefore, the joint probability that $x \notin S$ and $y$ is in a different component from $x$ in $G \sm S$ is at least $1/4 \cdot 1/2k = 1/8k$. Let a success be the event that $x \notin S$ and $y$ is in a different component from $x$. If we repeat the process of sampling two edges for $800k \log n$ times, the probability that no rounds succeed is at most
    \begin{align*}
        \left( 1 - \frac{1}{8k} \right)^{800k \log n} & \leq \left(\left( 1 - \frac{1}{8k} \right)^{8k} \right)^{100 \log n} \\
        & \leq \left( \frac{1}{e} \right)^{100 \log n} \\
        & \leq \left( \frac{1}{2} \right)^{100 \log n} \\
        & \leq n^{-100}.
    \end{align*}
    Hence, with probability $1 - n^{-100}$, we will eventually sample vertices $x$ and $y$ such that $x$ and $y$ are in different components of $G \sm S$. $\SHR(x, y)$ will list $S$ on line \ref{alg:balanced:call-SHR}. Furthermore, the algorithm only returns $k$-shredders because it only lists sets returned by $\SHR(\cdot, \cdot)$, which are guaranteed to be $k$-shredders.
    
    Now we calculate the time complexity. We sample two edges repeatedly for a total of $800k\log n$ rounds by line \ref{alg:balanced:num-samples}. For each round of sampling, we call $\SHR(\cdot, \cdot)$ once, which runs in $\O(km)$ time. We conclude that the entire algorithm runs in $\O(k \log n) \cdot \O(km) = \O(k^2 m \log n)$ time.
\end{proof}

\section{Omitted Proofs from Section \ref{sec:review}}
\label{sec:omitted-proofs:4}

\lemxystraddle*

\begin{proof}
    Let $S$ be a candidate $k$-shredder with respect to $\Pi$. Suppose that there exists a bridge $\Gamma$ of $\Pi$ that straddles $S$. We will show that there exists a path $P$ from $x$ to $y$ in $G \sm S$, which proves $S$ does not separate $x$ and $y$. Since $\Gamma$ straddles $S$, there exist paths $\pi_1, \pi_2$ in $\Pi$ and vertices $u \in \pi_1(\Gamma), v \in \pi_2(\Gamma)$ such that $\delta_{\pi_1}(u) < \delta_{\pi_1}(\pi_1(S))$ and $\delta_{\pi_2}(v) > \delta_{\pi_2}(\pi_2(S))$. First, we show that there exists a path $P$ from $u$ to $v$ in $G \sm S$. If $\Gamma$ is the edge $(u, v)$, the path $P = (u, v)$ exists in $G \sm S$ as $u \notin S$ and $v \notin S$. Otherwise, $\Gamma$ is a component $Q$ of $G \sm \Pi$ and both $u$ and $v$ are in $N(Q)$. This means there exist edges $(u, u')$ and $(v, v')$ in $G \sm S$ such that $u'$ and $v'$ are in $Q$. Because $Q$ is a component of $G \sm \Pi$, $v$ and $v'$ must be connected in $G \sm S$, as $S \subseteq \Pi$. Therefore, $u$ and $v$ are connected in $G \sm S$ via the vertices $u'$ and $v'$. Using $P$ as a subpath, we can construct a $x \e z$ path in $G \sm S$ as follows. Consider the path obtained by joining the intervals: $\pi_1[x \e u] \cup P \cup \pi_2[v \e z]$. Note that $\pi_1[x \e u]$ exists in $G \sm S$ because $\delta_{\pi_1}(u) < \delta_{\pi_1}(\pi_1(S))$, so no vertex of $S$ lives in the interval $\pi_1[x \e u]$. Similarly, $\pi_2[v \e z]$ exists in $G \sm S$ because $\delta_{\pi_2}(v) > \delta_{\pi_2}(\pi_2(S))$. Finally, as $P$ exists in $G \sm S$, the joined path from $x$ to $z$ exists in $G \sm S$.
    
    Conversely, suppose that no bridge of $\Gamma$ straddles $S$. We first show that there does not exist a path from $x$ to $y$ in $G \sm S$. Suppose for the sake of contradiction that there exists a path $P$ from $x$ to $y$ in $G \sm S$. Let $u$ denote the last vertex along $P$ such that $u \in \pi_1$ for some path $\pi_1 \in \Pi$ and $\delta_{\pi_1}(u) < \delta_{\pi_1}(\pi_1(S))$. Such a vertex must exist because $x \in P$ and satisfies $\delta_{\pi}(x) < \delta_{\pi}(\pi(S))$ for all paths $\pi \in \Pi$. Similarly, let $v$ denote the first vertex along $P$ after $u$ such that $v \in \pi_3$ for some path $\pi_3 \in \Pi$ and $\delta_{\pi_3}(v) > \delta_{\pi_3}(\pi_3(S))$. Such a vertex must exist because $y \in P$ and satisfies $\delta_{\pi}(\pi(S)) < \delta_{\pi}(y)$ for all paths $\pi \in \Pi$. Consider the open subpath $P(u \e v)$. We claim that every vertex in $P(u \e v)$ must not be in $\Pi$. To see this, suppose that there exists a vertex $w \in P(u \e v)$ such that $w \in \pi_3$ for a path $\pi_3 \in \Pi$. Then, $\delta_{\pi_3}(w) = \delta_{\pi_3}(\pi_3(S))$ as otherwise, this would contradict the definition of $u$ or $v$. This implies $w = \pi_3(S)$ because the vertex distance $\delta_{\pi_3}(\pi_3(S))$ from $x$ along $\pi_3$ is unique. However, this contradicts the definition of $P$ because $P$ is a path in $G \sm S$: it cannot contain $\pi_3(S)$. Hence, the inner vertices in the open subpath $P(u \e v)$ exist as a bridge $\Gamma$ of $\Pi$ with attachment set $\{u, v, \ldots \}$. If there are no vertices in the subpath, then $\Gamma$ is the edge $(u, v)$. Most importantly, we have that $\delta_{\pi_1}(u) < \delta_{\pi_1}(\pi_1(S))$ and $\delta_{\pi_3}(v) > \delta_{\pi_3}(\pi_3(S))$. Therefore, $\Gamma$ straddles $S$, arriving at our desired contradiction.
\end{proof}

\lemprune*

To prove \Cref{lem:prune}, we first show the following slightly different lemma. This lemma gives a workable data structure needed to prune candidate $k$-shredders efficiently.

\begin{lem} \label{lem:candidate-candidate}
    Let $x$ be a vertex in $G$ and let $\Pi$ denote a set of $k$ openly-disjoint paths starting from $x$ whose total length over all paths is at most $\ell$. Let $\C$ be a set of candidates along $\Pi$. There exists a deterministic algorithm that, given $(\Pi, \C)$ as input, lists all candidates in $\C$ which are not straddled by any other candidate in $\C$ in $\O(k |\C| \log |\C| + \ell)$ time.
\end{lem}

\begin{proof}
    The goal is to exploit the definition of straddling via sorting. First, order the $k$ paths in $\Pi$ arbitrarily. For each candidate $k$-shredder $S$ in $\C$, define the $k$-tuple $$(\delta_{\pi_1}(s_1), \delta_{\pi_2}(s_2), \dots, \delta_{\pi_k}(s_k))$$ where $s_i = \pi_i(S)$ is the vertex in $S$ that lies along the $i$-th path $\pi_i$ in $\Pi$. This $k$-tuple can be represented as a $k$-digit number in base $\ell$, inducing the lexicographical order between candidate $k$-shredders in $\C$. We sort the candidates by their $k$-tuple and organize the sorted list in a $|\C| \times k$ table (see \Cref{fig:sorted-tuples}) by non-decreasing order from top to bottom. This can be done via any stable sorting algorithm to sort each $k$-tuple for each of the $k$ digits, starting from the least significant digit to the most.
    \begin{figure}[!ht]
        \centering
        \includegraphics[scale=0.25]{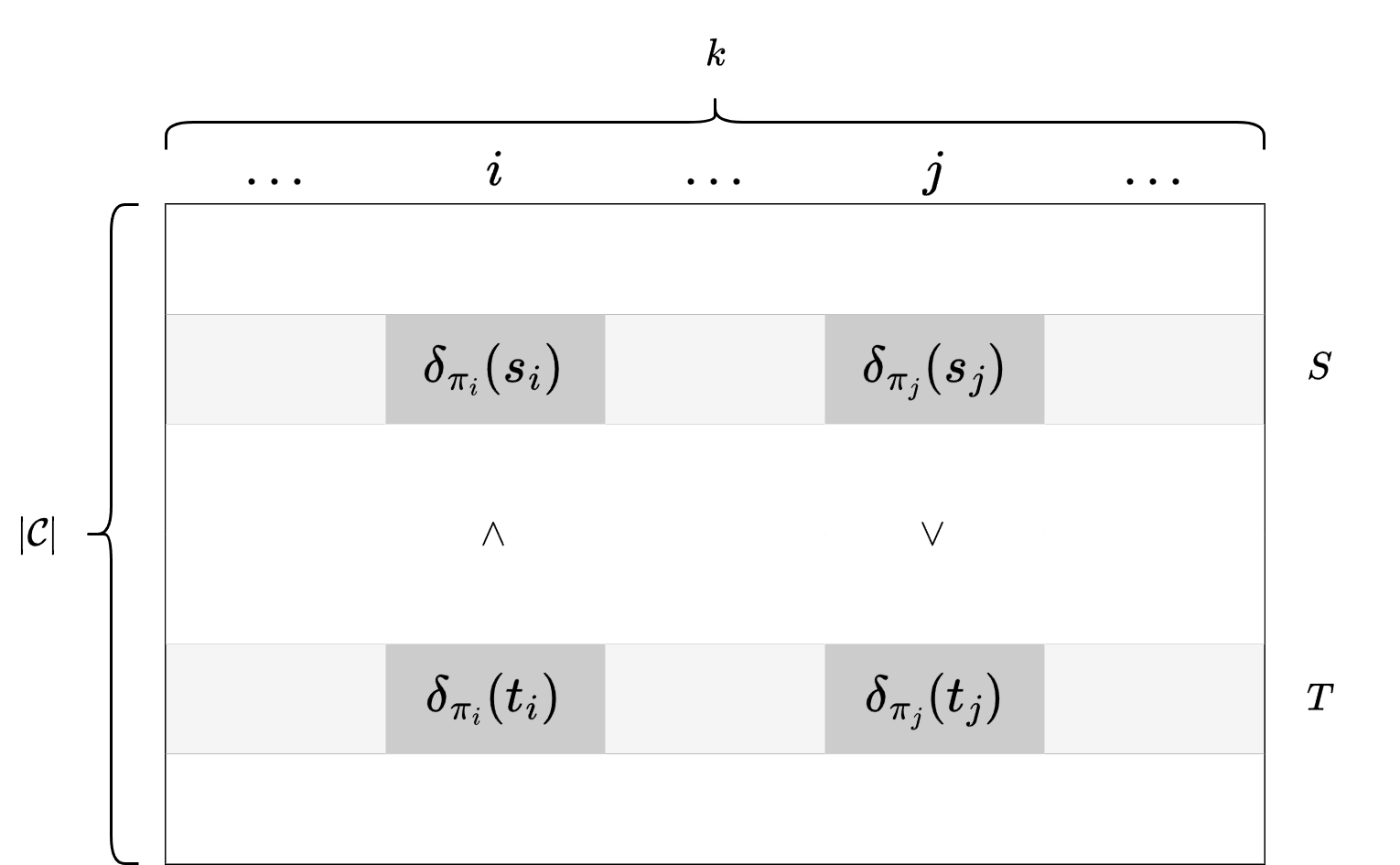}
        \caption{Candidate $k$-shredders sorted by their $k$-tuple numbers. Each row represents the $k$-tuple number $T_S$ for some candidate $k$-shredder $S \in \C$. Candidates are sorted in non-decreasing order from top to bottom. Here we have fixed two rows for two candidate $k$-shredders $S, T$ that straddle each other.}
        \label{fig:sorted-tuples}
    \end{figure}
    Using this sorted table, we can detect candidate $k$-shredders that are not straddled by others. The main idea is as follows. Suppose that a candidate $k$-shredder $S$ straddles another candidate $k$-shredder $T$. Let $s_i = \pi_i(S)$ and $ t_i = \pi_i(T)$. By \Cref{def:straddle}, there exist two paths $\pi_i, \pi_j \in \Pi$ (with $i < j$) such that $\delta_{\pi_i}(s_i) < \delta_{\pi_i}(t_i)$ and $\delta_{\pi_j}(s_j) > \delta_{\pi_j}(t_j)$. Without loss of generality, let $i$ be the smallest path index such that the above holds. Observe that because $i < j$ and $\delta_{\pi_i}(s_i) < \delta_{\pi_i}(t_i)$, the row for $S$ must be sorted above the row for $T$. On the other hand, in the $j$-th column of this table, we will see that the value in the row for $S$ ($\delta_{\pi_j}(s_j)$) is greater than the value in the row for $T$ ($\delta_{\pi_j}(t_j)$). This \textit{inversion} is the essential marker that we use to detect straddling. To formalize, a candidate $k$-shredder on row $s$ is straddled by another candidate if and only if there exists a column of the table with top-to-bottom-ordered entries $(e_1, \ldots, e_{|\C|})$ that contains an entry $e_t$ with $t < s$ and $e_t > e_s$, or there exists an entry $e_t$ with $t > s$ and $e_t < e_s$.

    Detecting which rows are involved in an inversion can be done as follows. For each column of the table, we scan through the column twice. In the first run, we scan from top to bottom, keeping track of the maximum value seen so far. If at any point, the entry on row $s$ is less than the maximum value seen so far, then we can infer that the $s$-th candidate $k$-shredder in the table was straddled. In the second run, we scan from bottom to top, keeping track of the minimum value seen so far. If the next entry on row $s$ is greater than the minimum, we can infer that the $s$-th candidate $k$-shredder in the table was straddled.

    We will prove why this scanning process suffices. Consider a candidate $k$-shredder sorted on row $s$. Suppose there exists a column of the table with top-to-bottom-ordered entries $(e_1, \ldots, e_{|\C|})$ such that there exists an entry $e_t$ with $t < s$ and $e_t > e_s$, or there exists an entry $e_t$ with $t > s$ and $e_t < e_s$. If $t < s$ and $e_t > e_s$, this implies that $e_s < \max\{e_1, \ldots, e_{s-1}\}$. We will detect this when scanning the column from top to bottom. Similarly, if $t > s$ and $e_t < e_s$, then $e_s > \min\{e_{s+1}, \ldots, e_{|\C|}\}$, which is detected when scanning the column from bottom to top. Conversely, suppose that when scanning a column from top to bottom, we detect that row $s$ is less than the maximum entry seen so far. Then, by definition, we have detected an inversion involving row $s$. When scanning from bottom to top, if we detect that row $s$ is greater than the minimum entry seen so far, we again have detected an inversion involving row $s$. Hence, this scanning process precisely captures any inversions if they occur. 
    
    Now, we calculate the time complexity. Firstly, we preprocess $\Pi$ so that for a vertex $v \in \Pi$, we can query the distance of $v$ from $x$ in constant time. This can be done in $\O(\ell)$ time by traversing each path in $\Pi$ starting from $x$ and marking the distance from $x$ for each traversed vertex. Next, we build the $k$-tuple numbers for each candidate in $\C$. This takes $\O(k|\C|)$ time. Since there are $|\C|$ $k$-tuple numbers, each of which are $k$ digits in base $\ell$, lexicographically sorting them by their digits takes $\O(k|\C| \log |\C|)$ time. The column scanning process can be completed in $\O(k|\C|)$ operations by iterating through the table twice. Hence, the entire algorithm runs in $\O(k|\C| \log |\C| + \ell)$ time.
\end{proof}

\begin{proof}[Proof of \Cref{lem:prune}]
    Firstly, we apply \Cref{lem:candidate-candidate} to eliminate candidate $k$-shredders that are straddled by other candidates in $\O(k|\C| \log |\C| + \ell)$ time. We are left with a set of bridges $\B$ and a set $\C$ of candidate $k$-shredders sorted by their $k$-tuple numbers in nondecreasing order. Because we have eliminated pairwise candidate straddling, the sorted list of remaining candidate $k$-shredders is well-ordered under $\preceq, \succeq$. That is, for any pair of candidate $k$-shredders $S_i, S_j$ ($i < j$), we have $S_i \preceq S_j$. Hence, if $S_j \preceq \Gamma$ for a bridge $\Gamma$ of $\Pi$, then $S_i \preceq \Gamma$ for all $i < j$. The goal is to develop a procedure that takes a bridge $\Gamma$ and determines all candidate $k$-shredders straddled by $\Gamma$ in $\O(\vol(\Gamma))$ time.

    Let $\pi \in \Pi$. Given a vertex $z \in \pi$, we can query for the greatest index $lo$ such that $S_{lo}$ satisfies $\delta_\pi(\pi(S_{lo})) \leq \delta_\pi(z)$. If no such index exists, then we set $lo = 0$. Similarly, we can query the smallest index $hi$ such that $S_{hi}$ satisfies $\delta_\pi(\pi(S_{hi})) \geq \delta_\pi(z)$. If no such index exists, we set $hi = |\C| + 1$. Using the sorted list of candidate $k$-shredders, we can preprocess the vertices of $\Pi$ in $\O(k |\C| + \ell)$ time to return $lo$ and $hi$ in constant time. To extend this a bit further, for each bridge $\Gamma \in \B$ we can compute the greatest index $lo$ and the least index $hi$ (if they exist) such that $S_{lo} \preceq \Gamma$ and $S_{hi} \succeq \Gamma$ in $\O(\vol(\Gamma))$ time. This can be done by simply scanning through the attachment set of $\Gamma$ and performing the queries described above. This is especially useful because every candidate $k$-shredder $S_i$ with $i \in (lo, hi)$ satisfies $S_i \not \preceq \Gamma$ and $S_i \not \succeq \Gamma$, which means $\Gamma$ straddles $S_i$. With these steps, we can compute an open interval $(lo, hi)$ containing the indices of all candidate $k$-shredders which are straddled by $\Gamma$. To find all candidate $k$-shredders which are straddled by any bridge in $\B$, we can simply merge all the $(lo, hi)$ intervals for each bridge $\Gamma \in \B$. Since $lo$ and $hi$ are nonnegative values that are on the order of $|\C|$, we can merge all intervals in $\O(|\C|)$ time. Tallying up all the steps, we have that the entire algorithm runs in $\O(k|\C| \log |\C| + \ell) + \O(\vol(\B)) + \O(|\C|) = \O(k|\C| \log |\C| + \ell + \vol(\B))$ time.
\end{proof}

\thmpairwise*

\begin{proof}[Proof of \Cref{thm:pairwise-shredders}]
    Given a pair of vertices $(x, y)$, obtaining a set $\Pi$ of $k$ openly-disjoint simple paths from $x$ to $y$ can be done in $\O(km)$ time. Finding all the bridges of $\Pi$ and obtaining the list of candidate $k$-shredders can be done in $\O(m+n)$ time using DFS/BFS. Note that the number of candidate $k$-shredders is $\O(n)$ because there are at most $\O(n)$ components of $G \sm \Pi$. 
    
    In \cite{CT99}, Cheriyan and Thurimella actually used radix sort to obtain the sorted table described in figure \Cref{fig:sorted-tuples}. Specifically, they noted that each digit in a $k$-tuple is bounded by $n$ because all simple paths from $x$ to $y$ are of length at most $n$. Hence, sorting at most $n$ $k$-digit numbers in base $n$ can be done in $\O(k (n + n)) = \O(kn)$ time using radix sort. After obtaining the sorted table, they followed the remaining steps in \Cref{lem:prune}. This entire pruning process therefore takes $\O(k(n + n) + m) = \O(kn + m)$ time. Hence, by \Cref{lem:xy-straddle}, we have obtained all $k$-shredders separating $x$ and $y$. In total, the algorithm takes $\O(km) + \O(m+n) + \O(kn + m) = \O(km)$ time.
\end{proof}

\section{Omitted Proofs from Section \ref{sec:local}}
\lemunverifiedset*
\begin{proof}
    We can assume that $|U| = k$ as otherwise $U = \{x\}$, which is caught by line \ref{alg:local:x-in-U}. Firstly, we have $\delta_{\pi}(\pi(U)) < \delta_{\pi}(z)$ because $U$ does not contain the endpoints of any path by line \ref{alg:local:U-prune}. The proof now becomes identical to that of \Cref{lem:local:no-straddle-is-shredder}.

    Assume for the sake of contradiction that $x$ and $z$ are connected in $G \sm U$. Then, there exists a path $P$ from $x$ to $z$ in $G \sm U$. Let $u$ denote the last vertex along $P$ such that $u \in \pi_1$ for some path $\pi_1 \in \Pi$ and $\delta_{\pi_1}(u) < \delta_{\pi_1}(\pi_1(U))$. Such a vertex must exist because $x \in P$ and satisfies $\delta_{\pi}(x) < \delta_{\pi}(\pi(U))$. Similarly, let $v$ denote the first vertex along $P$ after $u$ such that $v \in \pi_2$ for some path $\pi_2 \in \Pi$ and $\delta_{\pi_2}(v) > \delta_{\pi_2}(\pi_2(U))$. Such a vertex must exist because $z \in P$ and satisfies $\delta_{\pi}(\pi(U)) < \delta_{\pi}(z)$.

    Consider the open interval $P(u \e v)$ of $P$. We claim that every inner vertex $w$ in the open interval $P(u \e v)$ must not be in $\Pi$. To see this, suppose that there exists a vertex $w \in P(u \e v)$ such that $w \in \pi_3$ for a path $\pi_3 \in \Pi$. Then, $\delta_{\pi_3}(w) = \delta_{\pi_3}(\pi_3(U))$ as otherwise, we contradict the definition of $u$ or $v$. This implies $w = \pi_3(U)$ because the vertex distance $\delta_{\pi_3}(\pi_3(U))$ from $x$ along $\pi_3$ is unique. But this contradicts the definition of $P$ as $P$ is a path in $G \sm U$.

    Hence, the vertices of $P(u \e v)$ exist as a bridge $\Gamma$ of $\Pi$ with attachment set $\{u, v, \ldots \}$. If there are no vertices in $P(u \e v)$, then $\Gamma$ is the edge $(u, v)$. Furthermore, we have that $\delta_{\pi_1}(u) < \delta_{\pi_1}(\pi_1(U))$ and $\delta_{\pi_2}(v) > \delta_{\pi_2}(\pi_2(U))$. Therefore, $\Gamma$ straddles $U$. \Cref{lem:local:prune-straddled} implies $U$ would have been pruned. This is checked by line \ref{alg:local:U-prune}, which implies the algorithm would have returned $\varnothing$ for $U$, giving our desired contradiction.
\end{proof}

\lemtrackvolume*
\begin{proof}
    This was omitted from the pseudocode for simplicity. However, we can easily compute $\vol(Q_x)$ at the end of the algorithm by iterating through each path and exploring the bridges of $\Pi$ connected to the vertices preceding the unverified vertex. This takes $\O(\nu)$ time for each path because the unverified vertex $u$ is defined such the bridges attached to vertices preceding $u$ have total volume at most $\nu$. Hence, over all paths, this can be done in $\O(k \nu)$ time.
\end{proof}

\section{Proof of Lemma \ref{lem:hitting-set}}
\lemhittingset*
\begin{proof}
    If we sample an edge $(u, v)$ uniformly at random, we have
    \begin{alignat*}{2}
        & \Pr[u \in Q] && \geq \frac{1}{2} \cdot \frac{\vol(Q)}{m} \\
        \rightarrow & \Pr[u \notin Q] && \leq 1 - \frac{\vol(Q)}{2m}.
    \end{alignat*}
    If we independently sample $T$ edges $(u, v)$, we have
        $$
            \bigcap_{(u, v)} \Pr[u \notin Q] \leq \left(1 - \frac{\vol(Q)}{2m}\right)^T.
        $$
    If we set $T \geq \frac{2m}{\vol(Q)}$, then
        $$
            \left(1 - \frac{\vol(Q)}{2m}\right)^T \leq \left(1 - \frac{\vol(Q)}{2m}\right)^{\frac{2m}{\vol(Q)}} \leq \frac{1}{e}.
        $$
    Notice that the failure probability is now a constant which we can boost by scaling $T$ by a polylogarithmic factor. Because $\frac{1}{2} \nu < \vol(Q)$, we have
        $$\frac{2m}{\vol(Q)} \leq \frac{2m}{\frac{1}{2} \nu} \leq \frac{2m}{\frac{1}{2} \nu} \cdot 100 \log n = 400 \frac{m}{\nu} \log n.$$
    Setting $T = 400 \frac{m}{\nu} \log n$, we get
    \begin{align*}
        \bigcap_{(u, v)} \Pr[u \notin Q] \leq \left(1 - \frac{\vol(Q)}{2m}\right)^T = \left(1 - \frac{\vol(Q)}{2m}\right)^{\frac{2m}{\frac{1}{2} \nu} \cdot 100 \log n} & \leq \left(\frac{1}{e}\right)^{100 \log n} \\
        & \leq \left(\frac{1}{2}\right)^{100 \log n} \\
        & \leq n^{-100}.
    \end{align*}
    Therefore, with probability $1 - n^{-100}$, we will sample an edge $(u, v)$ such that $u \in Q$.
\end{proof}